\documentclass[11pt]{article}
\usepackage{fullpage}
\usepackage{amsmath,amsfonts,amsthm,amssymb}
\usepackage{url}

\usepackage{color}
\usepackage[usenames,dvipsnames,svgnames,table]{xcolor}
\usepackage[colorlinks=true, linkcolor=red, urlcolor=blue, citecolor=blue]{hyperref}

\usepackage{booktabs}
\usepackage{algorithm}
\usepackage{algpseudocode}
\usepackage{comment}
\usepackage{mathtools}
\usepackage{microtype}
\usepackage{caption}
\usepackage{float}
\usepackage{graphicx}
\usepackage{multirow}
\usepackage{mathrsfs}
\usepackage{newfloat}
\usepackage{relsize}

\newcommand{\CC}{C^*}

\newcommand{\N}{\mathbb{N}}
\newcommand{\R}{\mathbb{R}}

\newcommand{\OPT}{\mathcal{OPT}}

\newcommand{\ttx}[1]{\texttt{#1}}

\def\OPT{{\mathcal{OPT}}}
\def\expec#1#2{{\bf \mathbb{E}}_{#1}[ #2 ]}
\def\expecf#1#2{{\bf \mathbb{E}}_{#1}\left[ #2 \right]}

\numberwithin{equation}{section}
\numberwithin{figure}{section}
\theoremstyle{plain}
	\newtheorem{theorem}{Theorem}[section]

	\newtheorem{lemma}[theorem]{Lemma}

	\newtheorem{fact}[theorem]{Fact}

\theoremstyle{definition}
	\newtheorem{definition}[theorem]{Definition}
	
	\newtheorem*{remark*}{Remark}

\usepackage[framemethod=TikZ]{mdframed}
\newcounter{Frame}
\newenvironment{Frame}[1][htb]{%
\refstepcounter{Frame}
    \begin{mdframed}[%
        frametitle={#1},
        skipabove=\baselineskip plus 2pt minus 1pt,
        skipbelow=\baselineskip plus 2pt minus 1pt,
        linewidth=1.0pt,
        frametitlerule=true,
        nobreak=true
    ]%
}{%
    \end{mdframed}
}

\newenvironment{sproof}
{
 \par\noindent{\bfseries\upshape Proof sketch\ }
}
{\qed}

\title{Robust Communication-Optimal Distributed Clustering Algorithms\thanks{This work was supported in part by NSF grants CCF-1422910, CCF-1535967, IIS-1618714, an Office of Naval Research (ONR) grant N00014-18-1-2562, an Amazon Research Award, a Microsoft Research Faculty Fellowship, and a National Defense Science \& Engineering Graduate (NDSEG) fellowship. Part of this work was done while Ainesh Bakshi and David Woodruff were visiting the Simons Institute for the Theory of Computing.
This paper subsumes an earlier version of the paper by a subset of the authors \cite{balcan2017general}.
}} 

\author{
Pranjal Awasthi 
\\
Rutgers University\\
pranjal.awasthi@rutgers.edu
\and 
Ainesh Bakshi
\\
Carnegie Mellon University\\
abakshi@cs.cmu.edu
\and
Maria-Florina Balcan
\\
Carnegie Mellon University\\
ninamf@cs.cmu.edu
\and
Colin White
\\
Carnegie Mellon University\\
crwhite@cs.cmu.edu
\and David P. Woodruff
\\
Carnegie Mellon University\\
dwoodruf@cs.cmu.edu}

\begin{document}

\maketitle

\begin{abstract}

In this work, we study the $k$-median and $k$-means clustering problems 
when the data is distributed across many servers and can contain outliers. 
While there has been a lot of work on these problems for worst-case instances, 
we focus on gaining a finer understanding through the lens of beyond worst-case analysis. 
Our main motivation is the following: for many applications 
such as clustering proteins by function or clustering communities in a social network,
there is some unknown target clustering, 
and the hope is that running a $k$-median or $k$-means algorithm will produce clusterings which are close to matching the target clustering.
Worst-case results can guarantee constant factor approximations to the optimal $k$-median or $k$-means objective value, 
but not closeness to the target clustering.

Our first result is a distributed algorithm which returns a near-optimal clustering assuming a natural notion of stability, 
namely, \emph{approximation stability}~\cite{as}, even when a constant fraction of the data are outliers.
The communication complexity is $\tilde O(sk+z)$ where $s$ is the number of machines, $k$ is the number of clusters, and $z$ is the number of outliers.

Next, we show this amount of communication cannot be improved even in the setting when the input satisfies various non-worst-case assumptions. 
We give a matching $\Omega(sk+z)$ lower bound on the communication required both for approximating the optimal $k$-means or $k$-median cost up to any constant, and for returning a clustering that is close to the target clustering in Hamming distance.
These lower bounds hold even when the data satisfies approximation stability or other common notions of stability, and the cluster sizes are balanced. Therefore, $\Omega(sk+z)$ is a communication bottleneck, even for real-world instances. 

\end{abstract}

\clearpage

\section{Introduction} \label{sec:intro}
Clustering is a fundamental problem in machine learning with applications in many areas including computer vision, text analysis, bioinformatics, and so on. 
The underlying goal is to group a given set of points to maximize similarity inside a group
and dissimilarity among groups. A common approach to clustering is to set up an objective function and then approximately find the
optimal solution according to the objective. Common examples of these objective functions
include $k$-median and $k$-means, in which the goal is to find $k$ centers to minimize
the sum of the distances (or sum of the squared distances) from each point to its closest center.
Motivated by real-world constraints, further variants of clustering have been studied. For instance, in $k$-clustering with outliers,
the goal is to find the best clustering (according to one of the above objectives) 
after removing a specified number of data points, which is useful for noisy data.
Finding approximation algorithms to different clustering objectives and variants has attracted significant attention in the computer science community
\cite{arya2004local,byrka2015improved,charikar1999constant,outliers,chen2008constant,gonzalez1985clustering,makarychev2016bi}.

As datasets become larger, sequential algorithms designed to run on a single machine are no longer feasible for real-world applications. 
Additionally, in many cases data is naturally spread out among multiple locations.
For example, hospitals may keep records of their patients locally, but may want to cluster the
entire spread of patients across all hospitals in order to do better data analysis and inference.
Therefore, distributed clustering algorithms have gained popularity in recent years 
\cite{balcan2013distributed,distr_balanced,distr_kcenter, guha2017distributed, li2018distributed, cohen2015dimensionality, chen2016communication}. 
In the distributed setting, it is assumed that the data is partitioned arbitrarily across
$s$ machines, and the goal is to find a clustering which approximates the optimal
solution over the entire dataset while minimizing communication among machines.
Recent work in the theoretical machine learning community establishes guarantees on the
clusterings produced in distributed settings for certain problems 
\cite{balcan2013distributed,distr_balanced,distr_kcenter}.
For example, \cite{distr_kcenter} provides distributed algorithms for $k$-center and $k$-center with outliers,
and \cite{distr_balanced} introduces distributed algorithms for capacitated $k$-clustering under any $\ell_p$ objective. Along similar lines, the recent work of \cite{guha2017distributed} 
provides constant-factor approximation algorithms for $k$-median and $k$-means with $z$ outliers in the distributed setting. The work of Guha et al.\ also provides the best known communication complexity bounds for these settings and they scale as $O(sk+z)$ where $s$ is the number of machines, and $z$ is the number of outliers.

Although the above results provide a constant-factor approximation to $k$-median or $k$-means objectives,
many real-world applications desire a clustering that is close to a `ground truth' clustering in terms of the structure, 
i.e., the way the points are clustered rather than in terms of cost.
For example, for applications such as clustering proteins by function or clustering communities in a social network,
there is some unknown target clustering, and the hope is that running a $k$-median or $k$-means algorithm will produce 
clusterings which are close to matching the target clustering.
While in general having a constant factor approximation provides no guarantees on the closeness to the optimal clustering, a series of recent works has established that this is possible if the data has certain structural properties~\cite{awasthi2012center,awasthi2012improved,
as,balcan2012clustering,bilu2012stable,dick2017data,kumar2010clustering,voevodski2011min}.
For example, the \emph{$(1+\alpha,\epsilon)$-approximation stability} condition defined by \cite{as}
states that any $(1+\alpha)$-approximation to the clustering
objective is $\epsilon$-close to the target clustering. 
For such instances, it is indeed possible to output a clustering close to the ground truth in polynomial time, 
even for values of $\alpha$ such that computing a $(1+\alpha)$-approximation is NP-hard. 
We follow this line of research and ask whether distributed clustering is possible for non worst-case instances, in the presence of outliers. 
\vspace{-0.1in}
\subsection{Our contributions}
A distributed clustering instance consists of a set of $n$ points in a metric space partitioned arbitrarily
across $s$ machines. The problem is to optimize the $k$-median/$k$-means 
objective while minimizing the 
amount of communication across the machines. We consider algorithms that 
approximate the optimal cost as well as computing a clustering 
close to the target clustering in Hamming distance. Our contributions are as follows:
\begin{enumerate}
    \item In Section \ref{sec:aso}, we give a centralized clustering algorithm whose 
    output is $\epsilon$-close to the target clustering, in the presence of $z$ outliers,
    assuming the data satisfies $(1+\alpha,\epsilon)$-approximation stability and 
    assuming a lower bound on the size of the optimal clusters.
    To the best of our knowledge, this is the first polynomial time algorithm for clustering approximation stable instances in the presence of outliers. 
    Our results hold for arbitrary values of $z$, including when a
    constant fraction of the points are outliers, as long as there is a lower bound on the minimum cluster size.
    
    \item In Section \ref{sec:as}, we give a distributed algorithm whose output is close to the target clustering,
    assuming the data satisfies $(1+\alpha,\epsilon)$-approximation stability. 
    The communication complexity is 
    $\widetilde{O}\left(sk\right)$, where $s$ is the number of servers and $k$ is the number of clusters. In Section \ref{sec:distr}, we extend this to handle $z$ outliers, with a communication complexity $\widetilde{O}\left(sk+z\right)$.
    This matches the worst-case communication of \cite{guha2017distributed}, while outputting a near-optimal clustering by taking advantage of new 
    structural guarantees specific to approximation stability with outliers.
    \item While the above algorithms improve over worst-case distributed 
    clustering algorithms in terms of quality of the returned clustering,
    our algorithms use the same amount of communication as the worst case 
    protocols.
    In Section \ref{sec:lowerbounds}, we show that the $\Omega(sk)$ and 
    $\Omega(sk+z)$ communication costs for clustering without and with 
    outliers 
    are unavoidable even if data satisfies many types of stability assumptions
    that have been studied in the literature.
    Our lower bound of $\Omega(sk+z)$ for obtaining a $c$-approximation (for 
    any $c\geq 1$) holds even when the data is arbitrarily stable, 
    e.g., $(1+\alpha,\epsilon)$-approximation stable for all $\alpha\geq 0$ 
    and $0\leq\epsilon<1$.
    \item We also give an $\Omega(sk+z)$ lower bound for the problem of 
    computing a clustering whose Hamming distance is close to the optimal 
    clustering,
    even when the data is approximation-stable. Finally, we prove that our above $\Omega(sk + z)$ lower bounds hold for finding a clustering close to the 
    optimal in Hamming distance even when it is guaranteed that
    the optimal clusters are completely balanced, i.e., \ each cluster is of 
    size $\frac{n-z}{k}$ (in addition to the guarantee that the clustering 
    satisfies approximation stability), implying our algorithms from Section \ref{sec:aso} are optimal. Therefore, $\Omega(sk +z)$ is a fundamental communication bottleneck, even for real-world clustering instances. 

\end{enumerate}

\subsection{Related Work} \label{sec:related_work}
There is a long line of work on approximation algorithms for $k$-median and $k$-means clustering 
\cite{charikar1999constant,kanungo2002local,makarychev2016bi}, 
and the current best approximation ratios are 2.675 \cite{byrka2015improved} and 6.357 \cite{ahmadian2016better}, respectively.
The first constant-factor approximation algorithm for $k$-median with $z$ outliers was given by Chen \cite{chen2008constant},
and the current best approximation ratios for $k$-median and $k$-means  with outliers are $7.081+\epsilon$ and $53.002 + \epsilon$, respectively, given by Krishnaswamy et al.\ \cite{krishnaswamy2017constant}.
There is also a line of work on clustering with balance constraints on the clusters \cite{ahmadian2016approximation,aggarwal2006achieving,dick2017data}.
For $k$-median and $k$-means clustering in distributed settings, the work of Balcan et al.\ showed a coreset construction for $k$-median and
$k$-means, which leads to a clustering algorithm with $\tilde O(skd)$ communication, where $d$ is the dimension,
and also studied more general graph topologies for distributed computing \cite{balcan2013distributed}.
Malkomes et al.\ showed a distributed 13- and 4- approximation algorithm for $k$-center with and without
outliers, respectively \cite{distr_kcenter}.
Chen et al.\ studied clustering under the
broadcast model of distributed computing, and also proved a communication complexity lower bound
of $\Omega(sk)$ for distributed clustering \cite{chen2016communication}, building on a
recent lower bound for set-disjointness in the message-passing model \cite{braverman2013tight}.
Recently, \cite{guha2017distributed} showed a distributed algorithm with $\tilde O(sk+z)$ communication 
for computing a constant-factor approximation to $k$-median clustering with $z$ outliers. 
They also provide bicriteria approximations that remove $(1+\epsilon)z$ outliers to get a clustering of cost $O\left(1+\frac{1}{\epsilon}\right)$ times
the cost of the optimal $k$-median clustering with $z$ outliers, for any $\epsilon>0$. Even more recently, \cite{li2018distributed} showed that there exists a bi-criteria algorithm with communication independent of $z$ that achieves a constant approximation to the cost. In particular, their algorithm outputs $(1+\epsilon)z$ outliers and achieves a $(24 + \epsilon)$-approximation with $O\left(\frac{sk}{\epsilon} + \frac{s \log{\Delta}}{\epsilon} \right)$ communication, where $\Delta$ is the aspect ratio of the metric.

In recent years, there has also been a focused effort towards understanding clustering for non worst-case models
\cite{ostrovsky2012effectiveness,ackerman2009clusterability,bilu2012stable,kumar2010clustering}.
The work of Balcan et al.\ defined the notion of approximation stability and showed an algorithm which utilizes the structure to output a nearly
optimal clustering \cite{as}. 
Approximation stability has been studied in a wide range of contexts, including clustering \cite{symmetric,balcan2009agnostic,balcan2009finding}, the $k$-means$++$ heuristic \cite{agarwal2015k},
social networks \cite{Tim-social-net}, and computing Nash-equilibria \cite{awasthi2010nash}. 
A recent paper by Chekuri and Gupta introduces the model of clustering with outliers under perturbation resilience, a notion of stability which is
related to approximation stability \cite{chekuri2018perturbation}.

\section{Preliminaries}\label{sec:prelim}

Given a set $V$ of points of size $n$, a distance metric $d$, and an integer $k$,
let $\mathcal{C}$ denote a clustering of $V$, which we define as a partition of $V$ into $k$ subsets
$C_1,\dots,C_k$. 
Each cluster $C_i$ contains a center $c_i$.
When $d$ is an arbitrary distance metric, we must choose the centers from the point set.
If $V\subseteq\R^d$ and the distance metric is the standard Euclidean distance,
then the centers can be any $k$ points in $\R^d$.
In fact, this distinction only changes the cost of the optimal clustering by at most a factor of 2 by the triangle inequality for any $p$ (see, e.g., \cite{awasthi2014center}). 

The $k$-median, and the $k$-means costs are
$\sum_i\sum_{v\in C_i}d(c_i,v)$, and $\sum_i\sum_{v\in C_i}d(c_i,v)^2$ respectively.
For $k$ clustering with $z$ outliers, the problem is to compute the minimum cost clustering over $n-z$ points,
e.g., we must decide which $z$ points to remove, and how to cluster the remaining points, to minimize the cost.
We will denote the optimal $k$-clustering with $z$ outliers by $\mathcal{OPT}$, 
and we denote the set of outliers for $\OPT$ by $Z$.
We often overload notation and let $\mathcal{OPT}$ denote the objective value of the optimal clustering as well.
We denote the optimal clusters as $\CC_1,\dots, \CC_k$, with centers $c_1,\dots,c_k$.
We say that two clusterings $\mathcal{C}$ and $\mathcal{C}'$ are $\delta$-close if they
differ by only $\delta (n-z)$ points, i.e., $\min_\sigma\sum_{i=1}^k |C_i\setminus C_{\sigma(i)}'|<\delta (n-z)$. Let $\CC_{\min} = \min_{j\in[k]} |\CC_j|$, i.e., the minimum cluster size. Given a point $c \in V$, we define $V_c \subset V$ to be the closest set of $\CC_{\min}$ points to $c$.  

We study a notion of stability called approximation stability.
Intuitively, a clustering instance satisfies this assumption if 
all clusterings close in \emph{value} to 
$\mathcal{OPT}$ are also close in terms of the clusters themselves.
This is a desirable property when running an approximation algorithm, since in many applications,
the $k$-means or $k$-median costs are proxies for the final goal of recovering a clustering that is
close to the desired ``target'' clustering. Approximation stability makes this assumption explicit.
This was first defined for clustering with $z=0$~\cite{as}, however, we generalize the definition to the setting with outliers.
\begin{definition}(\emph{approximation stability.})
\label{def:alpha}
A clustering instance satisfies \emph{$(1+\alpha,\epsilon)$-approximation stability} for $k$-median or $k$-means with $z$ outliers if for
all $k$-clusterings with $z$ outliers, denoted by $\mathcal{C}$, if $\text{cost}(\mathcal{C})\leq (1+\alpha)\cdot\OPT$, then $\mathcal{C}$ is
$\epsilon$-close to $\mathcal{OPT}$.
\end{definition}

This definition implies that all clusterings close in cost to $\mathcal{OPT}$ must have nearly the same set of outliers.
This follows because if $\mathcal{C}$ contains more than $\epsilon (n-z)$ points from $Z$, then
$\mathcal{C}$ and $\mathcal{OPT}$ cannot be $\epsilon$-close.
This is similar to related models of stability for clustering with outliers, e.g.~\cite{chekuri2018perturbation}. 
Note it is standard in this line of work to assume the value of $\alpha$ is known~\cite{as}.

We will study distributed algorithms under the standard framework of the \emph{coordinator model}. There are $s$ servers, and a designated coordinator. Each server can send messages back and forth with the coordinator. This model is very similar to the \emph{message-passing model}, also known as the \emph{point-to-point} model, in which any pair of machines can send messages back and forth. In fact, the two models are equivalent up to constant factors in the communication complexity \cite{braverman2013tight}. Most of our algorithms can be applied to the mapreduce framework with a constant number of rounds. For more details, see \cite{distr_balanced,distr_kcenter}.

For our communication lower bounds, we work in the multi-party message passing model, 
where there are $s$ players, $P_1, P_2, \ldots, P_s$, who receive inputs $X^1$, $X^2$, \ldots $X^s$ respectively. They have access to private randomness 
as well as a common publicly shared random string $R$, and the objective 
is to communicate with a central coordinator who computes a function 
$f: X^1 \times X^2 \ldots \times X^s \to \{0, 1\} $ on the joint 
inputs of the players. The communication has multiple rounds and each player 
is allowed to send messages to the coordinator. Note, we can simulate communication 
between the players by blowing up the rounds by a factor of $2$. 
Given $X^i$ as an input to player $i$, let $\Pi\left(X^1, X^2, \ldots X^s\right)$ be 
the random variable that denotes the transcript between the players and the 
referee when they execute a protocol $\Pi$. For $i \in [s]$, let $\Pi_i$ 
denote the messages sent by $P_i$ to the referee. 

A protocol $\Pi$ is called a $\delta$-error protocol for function $f$ if there 
exists a function $\Pi_{out}$ such that for every input, 
$Pr\left[\Pi_{out}\left(\Pi(X^1, X^2, \ldots X^s)\right) = f(X^1, X^2, \ldots X^s)\right] \geq 1 - \delta$. 
The communication cost of a protocol, denoted by $|\Pi|$, is the maximum length of
$\Pi\left(X^1, X^2, \ldots, X^s\right)$ over all possible inputs and random coin 
flips of all the $s$ players and the referee. The randomized communication 
complexity of a function $f$, $R_{\delta}(f)$, is the communication cost 
of the best $\delta$-error protocol for computing $f$.

For our lower bounds, we also consider that the data satisfies a very strong, general notion of stability which we call $c$-separation.  
\begin{definition}(\emph{separation.})
Given, $c\geq 1$ and a clustering objective (such as $k$-means), a clustering instance satisfies \emph{$c$-separation} if 
$$c\cdot\max_i\max_{u,v\in \CC_i}d(u,v)<\min_i\min_{u'\in \CC_i,v'\notin \CC_i}d(u',v')$$
\end{definition}

Intuitively, this definition implies the maximum distance between any two points in one cluster is a factor $c$ smaller than the
minimum distance across clusters, as well as any clustering that
achieves a $(1+\alpha)$ approximation to the optimal cost must be
$\epsilon$ close to the target clustering in Hamming distance. Although this definition is quite strong, it has been used in several papers (for clustering with no outliers)
to show guarantees for various algorithms \cite{balcan2008discriminative,pruitt2011ncbi,kobren2017hierarchical}. 
We note that this notion of stability captures a wide class of previously studied notions including perturbation resilience \cite{bilu2012stable, awasthi2012center,balcan2012clustering, angelidakis2017algorithms} and approximation stability. 

\begin{definition}(\emph{perturbation resilience.})
For $\beta> 0 $ , a clustering instance $(V,d)$  satisfies $1+\alpha$-perturbation resilience for the $k$-means objective, if for any function $d': V \times V \to \mathbb{R}_{\geq 0}$, such that for all $p, q \in V$, $d(p,q) \leq d'(p,q) \leq (1+\beta)d(p,q)$, and the optimal clustering under $d'$ is unique and equal to the optimal clustering under $d$, for the $k$-means objective.
\end{definition}

We note we can replace the objective with any center based objective such as $k$-median or $k$-center. 
Next,
we show that \emph{separation} implies \emph{approximation stability} and
\emph{perturbation resilience}. 
We defer the proof to Appendix \ref{sec:separable}. 

\begin{lemma}
\label{lem:reductions}
Given $\alpha,\epsilon>0$, and a clustering objective (such as $k$-median), let $(V,d)$ be  
a clustering instance which satisfies $c$-separation, for $c>(1+\alpha)n$ (where $n=|V|$).
Then the clustering instance also satisfies $(1+\alpha, \epsilon)$-approximation stability and $(1+\alpha)$-perturbation resilience.
\end{lemma}

\section{Centralized Approximation Stability with Outliers} \label{sec:aso}

In this section, we give a centralized algorithm for clustering with $z$ outliers under approximation stability,
and then extend it to a distributed algorithm for the same problem.
To the best of our knowledge,
this is the first result for clustering with outliers under approximation stability, as well as the first distributed algorithm
for clustering under approximation stability even without outliers.

Our algorithm can handle any fraction of outliers, even when the set of outliers makes up
a constant fraction of the input points.
For simplicity, we focus on $k$-median. We show how to apply our result to $k$-means at the end of this section.

\begin{theorem}(Centralized Clustering.)
\label{thm:aso}
Algorithm \ref{alg:aso} runs in
poly$\left(n,\left(\frac{\alpha}{\epsilon}\left(k+\frac{1}{\alpha}\right)\right)^\frac{1}{\alpha}\right)$
time and outputs a clustering that is 
$\epsilon$-close to $\OPT$ for $k$-median with $z$ outliers under $(1+\alpha,\epsilon)$-approximation stability, 
assuming each optimal cluster $\CC_i$ has cardinality at least $2\left(1+\frac{5}{\alpha}\right)\epsilon(n-z)$.
\end{theorem}

Note that the runtime is at most poly$\left(n^\frac{1}{\alpha}\right)$, and if $\frac{\alpha}{\epsilon}\in\Theta(k)$, the runtime is poly$\left(n,k^\frac{1}{\alpha}\right)$.
The algorithm has two high-level steps. 
First, we use standard techniques from approximation stability without outliers to find a list of clusters $\mathcal{X}$, 
which contains clusters from the optimal solution (with $\leq\left(1+\frac{1}{\alpha}\right)\epsilon(n-z)$ mistakes),
and clusters made up mostly of outlier points.
We show how all but $1/\alpha$ of the outlier clusters must have high cost if their size were to be extended
to the minimum optimal cluster size, and can thus be removed from our list $\mathcal{X}$.
Finally, we use brute force enumeration to remove the final $\frac{1}{\alpha}$ outlier clusters, 
and after another cluster purifying step,
we are left with a $k$ clustering which $(1+\alpha)$-approximates the cost and thus is guaranteed to be $\epsilon$-close to optimal.

We begin by outlining the key properties of $(1+\alpha,\epsilon)$-approximation stability.
Let $w_{avg}$ denote the average distance from each point to its optimal center, 
so $w_{avg}\cdot (n-z)=\mathcal{OPT}$.
The following lemma is 
the first of its kind for clustering with outliers and establishes two key properties for approximation stable instances. Intuitively, the
first property bounds the number of points that are far away from
their optimal
center, and follows from Markov's inequality. The second property
bounds the number of points
that are either closer on average to the center of a non-optimal
cluster that the optimal one or
are outliers that are close to some optimal center as compared to a
point belonging to that cluster.

\begin{Frame}[\textbf{Algorithm \ref{alg:neighborhood}} : Computing the Neighborhood Graph]
\label{alg:neighborhood}
\ttx{Input}: Set of points $V$, parameters $\tau$, $b$
\begin{enumerate}
    \item Create the threshold graph $G_\tau=(V,E)$ by adding edge $(u,v)$ iff $d(u,v)\leq\tau$.
	\item Create graph $G'=(V,E')$ by adding edge $(u,v)$ iff $u$ and $v$ share $\geq b$ neighbors in $G_\tau$.
\end{enumerate}
\ttx{Output:}  Connected components of $G'$ 
\end{Frame}

\begin{lemma} \label{lem:struct}
Given a $(1+\alpha,\epsilon)$-approximation stable clustering instance $(V,d)$
for $k$-median such that for all $i$, $|C^*_i|>2\epsilon(n-z)$, then 
\begin{itemize}
    \item \textbf{Property 1:} For all $y>0$, there exist at most $\frac{y\epsilon}{\alpha}(n-z)$ 
points, $v$, such that $d(v,c_v)\geq\frac{\alpha w_{avg}}{y\epsilon}$.
    \item \textbf{Property 2:} There are fewer than $\epsilon (n-z)$ total points with one of the following two properties: the point $v$ is in an optimal cluster $C^*_i$, and there
exists $j\neq i$ such that $d(v,c_j)-d(v,c_i)\leq\frac{\alpha w_{avg}}{\epsilon}$, 
or, the point $v$ is in $Z$, and there exists $i$ and $v'\in C^*_i$ such that $d(v,c_i)\leq d(v',c_i)+\frac{\alpha w_{avg}}{\epsilon}$ (recall that $Z$ denotes the set of outliers from the optimal clustering).
\end{itemize}
\end{lemma}

\begin{proof}
Property 1 follows from Markov's inequality.
To prove property 2, assume the claim is false. 
Then there exists a set of points $V'\subseteq V\setminus Z$ such that each point $v \in V'$ is closer to a different center than its own center,
and a set of outlier points $Z'\subseteq Z$ such that each point $z \in Z'$  is close to some center, and $|V'\cup Z'|=\epsilon (n-z)$.
We define a new clustering $\mathcal{C}'$ by starting with $\mathcal{OPT}$ and making the following changes: 
each point $v\in V'$ moves to its second-closest center, and each point $z\in Z'$ joins its closest cluster, and then we
remove the $|Z'|$ points in $V\setminus V'\setminus Z$ which are furthest to their centers
(since all optimal clusters are size $>2\epsilon (n-z)$ and $|V'\cup Z'|=\epsilon (n-z)$, this is well-defined).
The cost increase of this new clustering will be at most 
$\frac{\alpha w_{avg}}{\epsilon}(\epsilon (n-z))\leq \alpha w_{avg} (n-z)$,
but it is not $\epsilon$-close to $\OPT$, causing a contradiction.
\end{proof}

We define a point as \emph{bad} if it falls into the bad case of either 
Property 1 (with $y=5$) or Property 2, and we denote the set of bad points by $B$.
Otherwise, a point is \emph{good}. From Properties 1 and 2, $|B|\leq \left(1+\frac{5}{\alpha}\right)\epsilon(n-z)$. 
For each $i$, let $G_i$ denote the good points from the optimal cluster $\CC_i$.
We then consider the graph $G'=(V,E')$ called the \emph{neighborhood graph}, 
constructed by adding an edge $(u,v)$ iff there are at least
$|B|+2$ points $w$ that  that are less than a threshold $\tau$, i.e., 
$d(u,w),d(v,w)\leq \tau=\frac{2w_{avg}}{5}$. 
Under approximation stability, 
the graph $G'$ has the following structure: there is an edge between all pairs 
of good points from $\CC_i$ and there is no edge between any pair of 
good points belonging to distinct clusters, $\CC_i, \CC_j$. Further, these points 
do not have any common neighbors. 
Since the set of good points in each cluster, denoted by $G_i$, form cliques of size $>|B|$ and are far away from one another, and there are $\leq |B|$ bad points,
it follows that each $G_i$ is in a unique connected component $C_i'$ of $G'$.

In the setting without outliers, the list of connected components of size greater than
$\left(1+\frac{5}{\alpha}\right)\epsilon n$ is exactly $\{C_1',\dots,C_k'\}$.
However, in the setting with outliers, we can only return a set $\mathcal{X}$ which includes $\{C_1',\dots,C_k'\}$ but also
may include many other outlier clusters which are hard to distinguish from the optimal clusters. 
Although approximation stability tells us that any set $Z'$ of outliers must have a much higher cost than any optimal cluster $\CC_i$ 
(since we can arrive at a contradiction by replacing the cluster $\CC_i$ with the cluster $Z'$),
this is not true when the size of $Z'$ is even slightly smaller than $\CC_i$. 
Since the good clusters returned are only $O\left(\frac{\epsilon}{\alpha}\right)$-close to optimal, many good clusters may be smaller than outlier clusters,
and so a key challenge is to distinguish outlier clusters $Z'$ from good clusters $C_i'$.

To accomplish this task, we compute the minimum cost of each cluster, pretending that its size is at least $\CC_{\min}$
(the size of the minimum optimal cluster, which we can guess in polynomial time).
In our key structural lemma (Lemma \ref{lem:badclique}), we show that nearly all outlier components will have large cost.
Given a set of points $Q$, we define $\text{cost}_{\min}(Q)$ to be the minimum cost of $Q$ if it were extended to $\CC_{\min}$ points. 
Note, $\text{cost}_{\min}(Q)$ can be computed in polynomial time by iterating over all points $c\in Q$, for each such point constructing $V_c$ by adding the the $\CC_{\min}-|Q|$ points closest to $c$, computing the resulting cost, and taking the minimum over all such costs.

\begin{lemma} \label{lem:badclique}
Given an instance of $k$-median clustering with $z$ outliers
such that each optimal cluster $|\CC_i|>2\left(1+\frac{5}{\alpha}\right)\epsilon(n-z)$, for any $x \in \mathbb{N}$, the instance satisfies $(1+\alpha,\epsilon)$-approximation stability for $\alpha>\frac{35}{5x-4}$, and 
 there are at most $x$ disjoint sets of outliers $Z'$ such that $|Z'|>\min_i |\CC_i|-\left(1+\frac{5}{\alpha}\right)\epsilon(n-z)$ and 
$\text{cost}_{\min}(Z')\leq \left(3+\frac{2\alpha}{5}\right)\frac{1}{x}\OPT$.
\end{lemma}

The key ideas behind the proof are as follows.
If there are two sets of outliers $Z_1$ and $Z_2$ both with fewer than $\CC_{\min}$ points, then we can obtain a contradiction
by taking into account both sets of outliers. Set $1\leq z_1,z_2\leq \left(1+\frac{5}{\alpha}\right)\epsilon(n-z)$ such that $|Z_1|=\CC_{\min}-z_1$ and
$|Z_2|=\CC_{\min}-z_2$, and assume without loss of generality that $z_1<z_2$. 
We design a different clustering $\mathcal{C}'$ by first replacing the minimum-sized cluster 
in the optimal clustering with $Z_1$.
The cost of the points in $Z_2$ is low by assumption. However, we have now potentially 
assigned more than $z$ points to be outliers by an additive $z_1$ amount. 

Hence, in order to create a valid clustering that is far from $\OPT$ we need to add back at 
least $z_1$ more outlier points. 
We do this by choosing $z_1$ outlier points from $Z_2$ that are closest to an optimal center 
in $\OPT$. 
To bound the additional cost incurred, we use the fact that $Z_2$ must be close to at least 
$z_2$ points from $V\setminus Z$, by the assumption that $\text{cost}_{\text{min}}(Z_2)$ is low,
and use these points to bound the distance from centers in $\OPT$ to the $z_1$ points that 
were added back.
In the full proof, we extend this idea to $x$ sets $Z_1,\dots,Z_x$ to achieve a tradeoff 
between $x$ and $\alpha$.

\vspace{0.1in}
\noindent \textit{Proof of Lemma \ref{lem:badclique}.}
Assume there are $x$ such disjoint sets of outliers, $Z_1,\dots,Z_x$ such that $|Z'|>\min_i |\CC_i|-\left(1+\frac{5}{\alpha}\right)\epsilon(n-z)$ and 
$\text{cost}_{\text{min}}(Z')\leq \left(3+\frac{2\alpha}{5}\right)\frac{1}{x}\OPT$.
First we show that for all $1\leq i\leq x$, $Z_i$ cannot contain  more than $\CC_{min}$ points.
Assume for sake of contradiction that $|Z_i|\geq \CC_{min}$. Then, there exists a center $c'\in Z_i$ such that
$\sum_{v\in Z_i}d(c',v)\leq\left(3+\frac{2\alpha}{5}\right)\frac{1}{x}\OPT$.
Then we arrive at a contradiction by replacing the minimum size optimal cluster with
$Z_i$, since the increase in cost is at most 
$$ \left(3+\frac{2\alpha}{5}\right)\frac{1}{x}\OPT<\alpha\cdot\OPT$$ (using $\alpha>\frac{35}{5x-4}$)
but the new clustering is not $\epsilon$-close to $\OPT$.

Now we can assume that all $Z_i$ contain fewer than $\CC_{min}$ points.
For all $1\leq i\leq x$, we denote $z_i=\CC_{min}-|Z_i|$, 
where $0<z_i<\left(1+\frac{5}{\alpha}\right)\epsilon(n-z)$. Recall, $V_c$ is the set of $\CC_{\min}$ closest points to $c$. 
Furthermore, denote $c_i'=\text{argmin}_{c\in Z_i}\sum_{v\in V_c}d(c,v)$ where $Z_i\subseteq V_c$ and $V_c\setminus Z_i$ 
contains the $\CC_{min}-|Z_i|$ closest points to $c$.
Then by assumption, 
$$\sum_{v\in V_{c_i'}}d(c_i',v)\leq \left(3+\frac{2\alpha}{5}\right)\frac{1}{x}\OPT.$$
Now given an arbitrary $1\leq i\leq x$, we modify $\OPT$ to create a new clustering $\mathcal{C}'$ as follows.
First we remove an arbitrary optimal cluster with size $\CC_{min}$ (by definition, such an optimal cluster must exist), 
then we add a new cluster $Z_i$ with center $c_i'$, and finally, we add 
the $z_i$ outliers closest to the current centers, to bring the size of the clustering back up to $n-z$.
Now we analyze the cost of this new clustering. We will show that for some $i$, the cost of this clustering is at most $(1+\alpha) \OPT$, 
contradicting approximation stability.
By assumption, we know that $$\sum_{v\in Z_i\cap Z} d(c_i',v)\leq \left(3+\frac{2\alpha}{5}\right)\frac{1}{x}\OPT,$$ 
so we only need to bound the cost of adding the $z_i$ next-closest outliers.
We set $j=i+1$ (or $j=1$ if $i=x$), and we consider the set $Z_j$. 
By assumption, $$\sum_{v\in V_{c_j'}}d(c_j',v)\leq 
\left(3+\frac{2\alpha}{5}\right)\frac{1}{x}\OPT.$$
Since $$Z_j\geq \CC_{min}-\left(1+\frac{5}{\alpha}\right)\epsilon(n-z)\geq\frac{1}{2}\cdot \CC_{min}$$ and $z_i=\CC_{min}-|Z_i|<\frac{1}{2}\cdot \CC_{min}$ 
there are at least $z_i$ non-outliers in $V_{c_j'}$. Call these points $V'_j$. 
Denote $\text{cost}(V'_j)=\sum_{v\in V'_j}d(v,c(v))$, where $c(v)$ denotes the center for $v$ in $\OPT$.
Also, we denote $\text{cost}'(V'_j)=\sum_{v\in V'_j}d(c_j',v)$ and $\text{cost}'(Z_j)=\sum_{v\in Z_j}d(c_j',v)$,
so $\text{cost}'(V'_j)+\text{cost}'(Z_j)\leq \left(3+\frac{2\alpha}{5}\right)\frac{1}{x}\OPT$.
Then by Markov's inequality, there must exist a point $v_j\in V'_j$ such that 
$$d(c(v_j),v_j)+d(v_j,c'_j)\leq \frac{1}{z_j}\left(\text{cost}'(V'_j)+\text{cost}(V'_j)\right)$$
Finally, the $z_i$ closest outliers in $Z_j$ to $z_i$ must have average cost at most $\frac{z_i}{z_j}\cdot\text{cost}'(Z_j)$.
Therefore, the cost of adding $z_i$ outliers to our clustering is at most 
\begin{equation*}
\frac{z_i}{z_j}\left(\text{cost}'(V'_j)+\text{cost}(V'_j)+\text{cost}'(Z_j)\right)
\leq \frac{z_i}{z_j}\left(\text{cost}(V'_j)+\left(3+\frac{2\alpha}{5}\right)\frac{1}{x}\OPT\right).
\end{equation*}
Now our goal is to show that for all valid settings of $z_1,\dots,z_x$ and
$\text{cost}(V'_1),\dots,\text{cost}(V'_x)$,
the maximum value of 
$$\min_{i\in [x]}\left(\frac{z_i}{z_j}\left(\text{cost}(V'_j)+\left(3+\frac{2\alpha}{5}\right)\frac{1}{x}\OPT\right)\right)$$
is at most $\left(3+\frac{2\alpha}{5}\right)\frac{1}{x}\OPT+\frac{1}{x}\cdot\OPT$.
Since $\sum_{\ell=1}^x\text{cost}(V'_\ell)\leq \OPT$, and $\prod_{\ell=1}^x\frac{z_i}{z_j}=1$,
we can solve to show the maximum value is when $z_1=\cdots=z_x$ and $\text{cost}(V'_1)=\cdots=\text{cost}(V'_x)=\frac{1}{x}\cdot\OPT$, and the minimum value over all $i\in [x]$ is
$$\left(3+\frac{2\alpha}{5}\right)\frac{1}{x}\OPT+\frac{1}{x}\cdot\OPT$$
Therefore, the total added cost for this clustering is
\begin{equation*}
\left(3+\frac{2\alpha}{5}\right)\frac{1}{x}\OPT+\left(3+\frac{2\alpha}{5}\right)\frac{1}{x}\OPT+\frac{1}{x}\cdot\OPT
\leq \left(7+\frac{4\alpha}{5}\right)\frac{1}{x}\OPT.
\end{equation*}
Since $\alpha>\frac{35}{5x-4}$, it follows that
$\left(7+\frac{4\alpha}{5}\right)\frac{1}{x}\OPT\leq \alpha\cdot\OPT$
Therefore, we have shown there exists a clustering which achieves cost $(1+\alpha)\OPT$ but is $\epsilon$-far from the optimal clustering, causing a contradiction.

\qed

\begin{Frame}[\textbf{Algorithm \ref{alg:aso}} : $k$-median with $z$-outliers under Approximation Stability]
\label{alg:aso}
\ttx{Input}: Clustering instance $(V,d)$, cost $w_{avg}$, value $\CC_{min}$, integer $x>0$.
\begin{enumerate}
    \item \label{step:neighborhood} Create the neighborhood graph on $V$ by running Algorithm \ref{alg:neighborhood} with parameters 
    $\tau=\frac{2 w_{avg}}{5\epsilon}$ and $b=\CC_{min}-(1+\frac{5}{\alpha})\epsilon (n-z)$ as follows:
    for each $u,v\in V$, add an edge $(u,v)$ iff there exist $\geq b$ points $w\in V$ such that $d(u,w),d(w,v)\leq \tau$.   
    Denote the connected components by $\mathcal{X}=\{Q_1,\dots,Q_d\}$.
	\item  For each $Q_i$, compute $\text{cost}_{\min}(Q_i)=\min_{c\in Q_i}\min_{V_c}\sum_{v\in V_c}d(c,v)$, where $V_c$ must satisfy $|V_c|\geq \CC_{min}$ and $Q_i\subseteq V_c$. Create a new set $\mathcal{X}'=\{Q_i\mid \text{cost}_{min}(Q_i)<\left(3+\frac{2\alpha}{5}\right)\frac{1}{x}\cdot\OPT \}$. \label{step:badclique}.
	\item For all $0\leq t\leq x$, for each size $t$ subset $\mathcal{X}'_t\subseteq\mathcal{X}'$ 
and size $\left(k-|\mathcal{X}'|-t\right)$ subset $\mathcal{X}_t\subseteq\left(\mathcal{X}\setminus\mathcal{X}'\right)$, \label{step:for}
\begin{enumerate}
\item Create a new clustering $\mathcal{C}=\mathcal{X}'\cup\mathcal{X}_t\setminus\mathcal{X}'_t$.
\item For each point $v\in V$, define $I(v)$ as the index of the cluster in $\mathcal{C}$ with minimum median distance to $v$,
e.g., $I(v)=\text{argmin}_i \left(d_{\text{med}}(v,Q_i)\right)$ where $d_{\text{med}}(v,Q_i)$ denotes the median distance from $v$ to $Q_i$.
\item Let $V'\subseteq V$ denote the $n-z$ points with the smallest values of $d(v,c_{I(v)})$. For all $i$, set $Q_i'=\{v\in V'\mid I(v)=i\}$.
\item If $\sum_i\text{cost}(Q_i')\leq (1+\alpha)\OPT$, return $\{Q_1,\dots,Q_k\}$.
\end{enumerate}
\end{enumerate}
\end{Frame}

From Lemma \ref{lem:badclique},
we show a threshold of $\text{cost}_{\min}$ for the components of $\mathcal{X}$,
such that all but $x$ optimal clusters are below the cost threshold, and all but $x$ outlier clusters are above the cost threshold.
Then we can brute force over all ways of excluding $x$ low-cost sets and including $x$ high-cost sets, and we will be guaranteed that
one combination contains a clustering which is $O\left(\frac{\epsilon}{\alpha}\right)$-close to the optimal.

However, we still need to recognize the right clustering when we see it.
To do this, we show that after performing one more cluster purifying step which is 
inspired by arguments in \cite{as}
 - reassigning all points to the component with the minimum median distance - we will reduce our error to $\epsilon(n-z)$ in Hamming distance
and we show how to bound the total cost of these mistakes by $\frac{4\alpha}{5}\OPT$.
Therefore, during the brute force enumeration, when we arrive at a clustering with cost at most $(1+\alpha)\OPT$, we return this clustering.
By definition of approximation stability, this clustering must be $\epsilon$-close to $\OPT$.

Since we are able to recognize the correct clustering (the one whose cost is at most $(1+\alpha)\OPT$), we can
try all possible values of $\CC_{min}$ while only incurring a polynomial increase in the runtime of the algorithm.
For computing $w_{avg}$, we first run an approximation algorithm for $k$-median with $z$ outliers to obtain a constant approximation to $w_{avg}$
(for example, we can use the 7.08-approximation for $k$-median with $z$ outliers \cite{krishnaswamy2017constant}).
The situation is much like the case where $w_{avg}$ is known, but the constant in the minimum allowed optimal cluster size increases by a factor of 7.
This is because we need to use a smaller value of $\tau$ when constructing the neighborhood graph $G'$, and so the number of ``bad'' points increases.
In order to show all the good connected components from $G'$ contain a majority of good points, we merely increase the bound on the minimum cluster size.

\vspace{0.1in}
\noindent \textit{Proof of Theorem \ref{thm:aso}.}
We start with the case where $w_{avg}$ and $\CC_{min}$ are known.
First, we show that after step \ref{step:neighborhood} of Algorithm \ref{alg:aso},
the set $\mathcal{X}$ contains $k$ clusters $C_i'$ such that $\{C_1',\dots,C_k'\}$ is $\left(1+\frac{5}{\alpha}\right)\epsilon(n-z)$-close to $\OPT$.

For each optimal cluster $\CC_i$, we define good points $X_i\subseteq \CC_i$ as follows:
a point $v\in X_i$ is good if it is not in the bad case of properties 1 (setting $y=5$) and 2
from Lemma \ref{lem:struct}. 
Then there are at most $\left(1+\frac{5}{\alpha}\right)\epsilon(n-z)$ bad points, and at most $\epsilon (n-z)$ of the bad points are in $Z$.
Recall the conditions from the threshold graph $G_\tau$:
\emph{(1)} For all $i$, for all $u,v\in X_i$, $(u,v)\in E(G_{\tau})$. 
\emph{(2)} For $u\in X_i$ and $v\in X_{j\neq i}$, $(u,v)\notin E(G_{\tau})$,
furthermore, these points do not share any common neighbors in $G_{\tau}$.
Therefore, each $X_i$ is a clique in $G_{\tau}$, with no common neighbors to the other cliques.

From Lemma \ref{lem:struct}, we also have that at most $\epsilon (n-z)$ total outliers have a neighbor to any good point. 
Call these the ``bad outliers''.
This implies that at most $\epsilon(n-z)$ outliers share 
$\geq \CC_{min}-\left(1+\frac{5}{\alpha}\right)\epsilon(n-z)$ neighbors
with a good point: the only common neighbors can be bad points and bad outliers,
which is $<\left(1+\frac{5}{\alpha}\right)\epsilon(n-z)$.
It follows that for all $i$, there is a component $C_i'$ in $G'$ which is close to $\CC_i$, formally,
the set of clusters $\{C_1',\dots,C_k'\}$ is 
$\left(1+\frac{5}{\alpha}\right)\epsilon(n-z)$-close to $\OPT$, 
where the error comes from bad points and bad outliers.
Note that every erroneous point is still at most
$\frac{2\alpha w_{avg}}{5\epsilon}$ from its center.
Then we have 
\begin{align*}
\text{cost}(\{C_1',\dots,C_k'\})&\leq \OPT+\frac{2\alpha w_{avg}}{5\epsilon}\left(1+\frac{5}{\alpha}\right)\epsilon(n-z)\\
&\leq \left(3+\frac{2\alpha}{5}\right)\OPT.
\end{align*}
By a Markov inequality, at most $x$ clusters in $\{C_1',\dots,C_k'\}$ have cost greater than $\left(3+\frac{2\alpha}{5}\right)\frac{1}{x}\OPT$.
The rest of the graph $G'$ consists of outliers and up to $\left(1+\frac{5}{\alpha}\right)\epsilon(n-z)$ bad points from $V\setminus Z$ 
which can make up small or large components.
From Lemma \ref{lem:badclique}, at most $x$ of these components have cost less than or equal to $\left(3+\frac{2\alpha}{5}\right)\frac{1}{x}\OPT$.
Therefore, after step \ref{step:badclique} of Algorithm \ref{alg:aso}, $\mathcal{X}'$ contains at least $k-x$ good clusters.
Then there exists a step of the for loop in step \ref{step:for} such that $\mathcal{C}=\{C_1',\dots,C_k'\}$.
We will show that the algorithm returns a clustering that is $\epsilon$-close to $\OPT$.

Consider the step of the for loop such that $\mathcal{C}=\{C_1',\dots,C_k'\}$.
We show how step \ref{step:for} of Algorithm \ref{alg:aso}
brings the error down from $\left(1+\frac{5}{\alpha}\right)\epsilon(n-z)$ to $\epsilon(n-z)$.
Consider a point $v\in V\setminus Z$
which is not in the bad case of Property 2 of 
Lemma \ref{lem:struct}, specifically, $v$ is in an optimal cluster $\CC_i$ such that for all $j\neq i$,
we have $d(v,c_j)-d(v,c_i)>\frac{\alpha w_{avg}}{\epsilon}$.
Given good points $x\in X_i$ and $y\in X_{j\neq i}$,
we have 
\begin{align*}
d(v,x)&\leq d(v,c_i)+d(c_i,x)\\
&\leq d(v,c_j)-\frac{\alpha w_{avg}}{\epsilon}+\frac{\alpha w_{avg}}{5\epsilon}\\
&\leq d(v,y)-d(y,c_j)-\frac{4\alpha w_{avg}}{5\epsilon}\\
&\leq d(v,y)-\frac{3\alpha w_{avg}}{5\epsilon}.
\end{align*}
Since there are fewer than $\left(1+\frac{5}{\alpha}\right)\epsilon(n-z)$
total errors in $\{C_1',\dots,C_k'\}$, and for all $i$, 
$|C_i|>2\left(1+\frac{5}{\alpha}\right)\epsilon(n-z)$,
it follows that the majority of points in $C_i'$ are good points.
Therefore, for all $j\neq i$, we have $d_{\text{med}}(v,C_i')+\frac{3\alpha w_{avg}}{5\epsilon}
<d_{\text{med}}(v,C_j')$ (recall that $d_{\text{med}}$ denotes the median distance from $v$ to $Q_i$).

If we look at all points in $V\setminus Z$, the clustering created using $I(v)$ will have $\epsilon(n-z)$ errors.
Whenever a point is misclustered, e.g., a point $v\in \CC_i$ is put into cluster $\CC_j$, we must have $d(v,c_j)<d(v,c_i)+\frac{2 w_{avg}}{5\epsilon}$,
so the additive increase in cost to the clustering is at most $\frac{2\alpha w_{avg}}{5}$.
It is possible that some outlier points $z\in Z$ will have a smaller value of $d_{\text{med}}(z,c_{I(z)})$ than a point $v\in V\setminus Z$,
but this can only happen for $\epsilon(n-z)$ pairs $(z,v)$ due to Lemma \ref{lem:struct}.
Again, this type of mistake can only add $\frac{2\alpha w_{avg}}{5}$ to the total cost of the clustering, since $d(z,c(v))<d(v,c(v))$.
Therefore, we have 
\begin{align*}
\text{cost}(Q_1',\dots Q_k')&\leq \OPT+\frac{2 w_{avg}}{5}\cdot 2\epsilon(n-z)\\
&\leq (1+\alpha)\OPT.
\end{align*}
By definition of approximation stability, this clustering must be $\epsilon$-close to $\OPT$.

Now we move to the case where $w_{avg}$ and $\CC_{min}$ are not known.
For $w_{avg}$, we run an approximation algorithm for $k$-median with $z$-outliers to
obtain a constant approximation to $w_{avg}$ 
(for example, there is a recent 7.08-approximation for $k$-median with $z$ outliers \cite{krishnaswamy2017constant}).
The situation is much like the case where $w_{avg}$ is known,
but the constant in the minimum allowed optimal cluster size increases by a factor of 7. 
The algorithm proceeds the same way as before.
If $\CC_{min}$ is not known, we can run the algorithm for $\hat C=n,n-1,n-2$, etc., until step \ref{step:for}
returns a clustering with cost $\leq (1+\alpha)w_{avg}(n-z)$, at which point we are guaranteed that
the clustering is $\epsilon$-close to $\OPT$.
Step 3 
searches through at most $x\cdot {k \choose x}\cdot {n \choose x}$ tuples,
and all other steps in Algorithm \ref{alg:aso} are polynomial in $n$. This completes the proof.

\qed

\section{Distributed Approximation Stability without Outliers} \label{sec:as}
In this section, we give the first distributed algorithms for approximation stability when there are no outliers.
We present two algorithms that use $\widetilde{O}(sk)$ communication to output near-optimal clusterings of the input points.
The first theorem outputs an $O\left(\left(1+\frac{1}{\alpha}\right)\epsilon\right)$-close clustering with no assumptions other than approximation stability,
and the next theorem outputs an $O(\epsilon)$-close clustering assuming the optimal clusters are large. The lower bounds presented
in Section \ref{sec:lowerbounds} imply that the algorithms are communication optimal. 

\begin{Frame}[\textbf{Algorithm \ref{alg:iterative_greedy}} : Iterative Greedy Procedure]
\label{alg:iterative_greedy}
\ttx{Input}: Set of points $V$, parameters $\tau$, $k$
\begin{enumerate}
    \item Create the threshold graph $G_\tau=(V,E)$ by adding edge $(u,v)$ iff $d(u,v)\leq\tau$.
    \item Initialize $A=\emptyset$, $V'=V$. For all $v$, let $N(v)=\{u\mid (u,v)\in E\}$.
    \item While $|A|<k$,  \label{step:mii}
    set $v'=\text{argmax}_{v\in V'}N(v)\cap V'$, and define $C(v')=N(v')\cap V'$.
    \begin{enumerate}
        \item Add $(v',C(v'))$ to $A$, and remove $N(v')$ from $V'$.
    \end{enumerate}
\end{enumerate}
\ttx{Output:}  Center and cluster pairs $A=\{(v_1,C(v_1)),\dots,(v_k,C(v_k))\}$
\end{Frame}

\begin{theorem} \label{thm:distr_as}
Given a $(1+\alpha,\epsilon)$-approximation stable clustering instance, with high probability, Algorithm $\ref{alg:distr_as}$ outputs a clustering that is $O\left(\epsilon\left(1+\frac{1}{\alpha}\right)\right)$-close to $\OPT$ for $k$-median under 
$(1+\alpha,\epsilon)$-approximation stability with $\widetilde{O}(sk)$ communication.
\end{theorem}

We achieve a similar result for $k$-means. 
We also show that if the optimal clusters are large, the error of the outputted clustering can be pushed even lower.

\begin{theorem} \label{thm:distr_as_large}
There exists an algorithm which outputs a clustering that is 
$O(\epsilon)$-close to $\OPT$ for $k$-median under 
$(1+\alpha,\epsilon)$-approximation stability with $O(sk\log n)$ communication
if each optimal cluster $\CC_i$ has size $\Omega\left(\left(1+\frac{1}{\alpha}\right)\epsilon n\right)$.
\end{theorem}

First we explain the intuition behind Theorem \ref{thm:distr_as}.
The high level structure of the algorithm can be thought of as a two-round version of
Algorithm \ref{alg:iterative_greedy}: first each machine clusters its local point set using
Algorithm \ref{alg:iterative_greedy}, and sends the weighted centers to the coordinator.
The coordinator runs Algorithm \ref{alg:iterative_greedy} on the weighted centers, using a higher threshold
value, to output the final solution.

\begin{Frame}[\textbf{Algorithm \ref{alg:distr_as}} : Distributed $k$-median clustering under 
$(1+\alpha,\epsilon)$-approximation stability]
\label{alg:distr_as}
\ttx{Input}: Distributed points $V=V_1\cup\cdots\cup V_m$, average $k$-median cost $w_{avg}$
\begin{enumerate}
    \item For each machine $i$,
        \begin{itemize}
            \item \label{step:mi} Run Algorithm \ref{alg:iterative_greedy} with $\tau=\frac{\alpha w_{avg}}{9\epsilon}$ , 
            outputting $A'_i=\{(v_1^i,C(v_1^i)),\dots,(v_k^i,C(v_k^i))\}$.
            \item Send $A_i=\{(v_1^i,|C(v_1^i)|),\dots,(v_k^i,|C(v_k^i)|)\}$ to the coordinator. 
        \end{itemize}
    \item Given the set of weighted points received, $A=\cup_i A_i$,
    the coordinator runs  Algorithm \ref{alg:iterative_greedy} with graph $\tau'=3\tau$ and the weighted points $A$, outputting $$G'=\left\{(x_1,C(x_1)),\dots,(x_k,C(x_k))\right\}$$ \label{step:m1}
\end{enumerate}
\vspace{-0.3in}
\ttx{Output:} Centers $G=\{x_1,\dots,x_k\}$
\end{Frame}

\begin{lemma} \label{lem:greedy} \cite{as}
\footnote{This lemma is obtained by merging Lemma 3.6 and Theorem 3.9 from \cite{as}.}
Given a graph $G$ over good clusters $G_1,\dots G_k$ and bad points $B$, with the following properties:
\begin{enumerate}
\item For all $u,v \in G_i$, edge $(u,v)$ is in $E(G)$.
\item For $u\in G_i$, $v\in G_j$ such that $i\neq j$, then $(u,v)\notin E(G)$, moreover,
$u$ and $v$ do not share a common neighbor in $G$.
\end{enumerate}
Then let $C(v_1),\dots,C(v_k)$ denote the output of running Algorithm \ref{alg:iterative_greedy} on $G$
with parameter $k$. There exists a bijection
$\sigma:[k]\rightarrow[k]$ between the clusters $C(v_i)$ and $G_j$ such that
$\sum_i |G_{\sigma(i)}\setminus C(v_i)|\leq 3|B|$.
\end{lemma}

\begin{proof} 
From the first assumption, each good cluster $G_i$ is a clique in $G$. Initially,
let each clique $G_i$ be ``unmarked'', and then we ``mark'' it the first time the
algorithm picks a $C(v_j)$ that intersects $G_i$. A cluster $C(v_j)$ can
intersect at most one $G_i$ because of the second assumption.
During the algorithm, there will be two cases to consider. 
If the cluster $C(v_j)$ intersects an unmarked clique $G_i$,
then set $\sigma(j)=i$.
Denote $|G_i\setminus C(V_j)|=r_j$. Since the algorithm chose the maximum degree node
and $G_i$ is a clique, then there must be at least $r_j$ points from $B$ in $C(V_j)$.
So for all cliques $G_i$ corresponding to the first case, we have
$\sum_j |G_{\sigma(j)}\setminus C(v_j)|\leq\sum_j r_j\leq |B|$.

If the cluster $C(v_j)$ intersects a marked clique, then assign $\sigma(j)$ to an
arbitrary $G_{i'}$ that is not marked by the end of the algorithm.
The total number of points in all such $C(v_j)$'s is at most the number of points
remaining from the marked cliques, which we previously bounded by $|B|$, plus up
to $|B|$ more points from the bad points.
Because the algorithm chose the highest degree nodes in each step,
each $G_{i'}$ has size at most the size of its corresponding $C(v_j)$. 
Therefore, for all cliques $G_{i'}$ corresponding to the second case, we have
$\sum_j |G_{\sigma(j)}\setminus C(v_j)|\leq\sum_j |G_{\sigma(j)}|\leq 2|B|$.
Thus, over both cases, we reach a total error of $3|B|$.
\end{proof}

Our proofs crucially use the structure outlined in Lemma \ref{lem:struct},
as well as properties \emph{(1)} and \emph{(2)} about the threshold graph $G_{\tau}$
from Section \ref{sec:aso}.

\begin{proof} \textbf{(Theorem \ref{thm:distr_as})}
The proof is split into two parts, both of which utilize Lemma \ref{lem:greedy}. 
First, given machine $i$ and $1\leq j\leq k$, let $G_j^i$ denote
the set of good points from cluster $\CC_j$ on machine $i$. Let $B_i$ denote the set of bad points on machine
$i$. 
Given $u,v\in G_j^i$, $d(u,v)\leq d(u,c_j)+d(c_j,v)\leq 2t$, so $G_j^i$ is a clique in $G_{2t}^i$.
Given $u\in G_j^i$ and $v\in G_{j'}^i$ such that $j\neq j'$, then
\begin{align*}
d(u,v)>d(u,c_{j'})-d(c_{j'},v)\geq 18t-d(u,c_j)-d(c_{j'},v)>16t.
\end{align*}
Therefore, if $u$ and $v$ had a common neighbor $w$ in $G_{2t}^i$, $$16t<d(u,v)\leq d(u,w)+d(v,w)\leq 4t$$
causing a contradiction.
Since $G_{2t}^i$ satisfies the conditions of Lemma \ref{lem:greedy}, it follows that there
exists a bijection $\sigma:[k]\rightarrow[k]$ 
between the clusters $C(v_j)$ and the good clusters $G_\ell$ such that
$\sum_j |G_{\sigma(j)}^i\setminus C(v_j)|\leq 3|B_i|$.
Therefore, all but $3|B_i|$ good points on machine $i$ are within $2t$ of some point in $A_i$.
Across all machines, $\sum_i |B_i|\leq |B|$, so there are less than $4|B|$ good points which are not
distance $2t$ to some point in $A$.

Since two points $u\in G_i$, $v\in G_j$ for $i\neq j$ are distance $>16t$, then each point in $A$
is distance $\leq 2t$ from good points in at most one set $G_i$. Then we can partition $A$ into sets
$G_1^A,\dots, G_k^A,B'$, such that for each point $u\in G_i^A$, there exists a point $v\in G_i$ such that
$d(u,v)\leq 2t$. The set $B'$ consists of points which are not $2t$ from any good point.
From the previous paragraph, $|B'|\leq 3|B|$, where $|B'|$ denotes the sum of the weights of all points in $B'$.
Now, given $u,v\in G_i^A$, there exist $u',v'\in G_i$ such that $d(u,u')\leq 2t$ and $d(v,v')\leq 2t$, 
and $d(u,v)\leq d(u,u')+d(u'c_i)+d(c_i,v')+d(v',v)\leq 6t$
\begin{equation*}
\begin{split}
d(u,v)\leq d(u,u')+d(u'c_i)+d(c_i,v')+d(v',v)\leq 6t
\end{split}
\end{equation*}
Given $u\in G_i^A$ and $w\in G_j^A$ for $i\neq j$, there exist $u'\in G_i$, $w'\in G_j$ such that
$d(u,u')\leq 2t$ and $d(w,w')\leq 2t$.
\begin{equation*}
\begin{split}
d(u,w) & \geq d(u',c_j)-d(u,u')-d(c_j,w')-d(w,w')\\
& >(18t-d(u,c_i))-2t-t-2t\\ 
& \geq 12t.
\end{split}
\end{equation*}

Therefore, if $u$ and $w$ had a common neighbor $w$ in $G_{6t}$, then $12t<d(u,v)\leq d(u,w)+d(v,w)\leq 12t$,
causing a contradiction.
Since $G_{6t}$ satisfies the conditions of Lemma \ref{lem:greedy} it follows that there
exists a bijection $\sigma:[k]\rightarrow[k]$ 
between the clusters $C(v_i)$ and the good clusters $G_j^A$ such that
$\sum_j |G_{\sigma(j)}^A\setminus C(v_j)|\leq 3|B'|$.
Recall the centers chosen by the algorithm are labeled as the set $G$. Let $x_i\in G$ denote the center
for the cluster $G_i$ according to $\sigma$.
Then all but $3|B'|$ good points $u\in G_i$ are distance $2t$ to a point in $A$ which is distance $6t$
to $x_i$. $u$ must be distance $>8t$ to all other points in $G$ because they are distance $2t$ from good
points in other clusters.
Therefore, all but $3|B'|\leq 12|B|$ good points are correctly clustered.
The total error over good and bad points is then $12|B|+|B|=13|B|\leq (48+\frac{468}{\alpha})\epsilon n$
so the algorithm achieves error $O(\epsilon(1+\frac{1}{\alpha}))$.
There are $sk$ points communicated to the coordinator, the weights can be represented by $O(\log n)$ bits, so the total communication
is $\widetilde{O}(sk)$. This completes the proof for $k$-median when the
algorithm knows $w_{avg}$ up front.

When Algorithm \ref{alg:distr_as} does not know $w_{avg}$, then it first runs 
a worst-case approximation algorithm to obtain an estimate $\hat w\in [w_{avg},\beta w_{avg}]$
for $\beta\in O(1)$. 
Now we reset $t$ in Algorithm \ref{alg:distr_as} to be 
$\hat t=\frac{\alpha\beta w_{avg}}{18\epsilon}$.
Then the set of bad points grows by a factor of $\beta$, but the same analysis still holds,
in particular, Lemma \ref{lem:greedy} and the above paragraphs go through, adding a factor
of $\beta$ to the error and only increases communication by a constant factor.

\end{proof}

The key ideas behind the proof of Theorem \ref{thm:distr_as_large} are as follows.
First, we run Algorithm \ref{alg:distr_as} to output a clustering with error $O\left(\left(1+\frac{1}{\alpha}\right)\epsilon\right)$.
To ensure $O(\epsilon)$ error when further assuming the optimal clusters are large,
we can use a technique similar to the one in the previous section: for each unassigned point $v$, 
assign this point to the cluster with the minimum median distance to $v$.
The key challenge is to run this technique without using too much communication, since we cannot send the entire set $A$ (which is size $\Theta(sk)$) to each machine.
To reduce the communication complexity, we instead randomly sample $\Theta\left(\frac{\log k}{\epsilon'}\right)$ points from $A$ and
send each to machine $i$, incurring a communication cost of $O\left(\frac{s\log(k)}{\epsilon'}\right)$. Note, the $\epsilon'$ is not the stability parameter, but used to obtain a point that is a $1+\epsilon'$ approximation to center of each cluster. 
Now each point $v\in V$ calculates the index of the cluster with the minimum median distance to $v$, over the sample.
Using a Chernoff bound, we show that for each point $v$ and each cluster $C_i$, the median of the sampled points must come from the core of $\CC_i$,
ensuring that $v$ is correctly classified.

\begin{proof} \textbf{(Theorem \ref{thm:distr_as_large})}
The algorithm is as follows. First, run Algorithm \ref{alg:distr_as}.
Then send $G'$ to each machine $i$, incurring a communication cost of $O(sk)$.
For each machine $i$, for every point $v\in V_i$, calculate the median distance from
$v$ to each cluster $C(x_j)$ (using the weights). Assign $v$ to the index $j$ with
the minimum median distance. Once every point undergoes this procedure, 
call the new clusters $G_1,\dots,G_k$, where $G_j$ consists of all points
assigned to index $j$. Now we will prove the clustering $\{G_1,\dots,G_k\}$
is $O(\epsilon)$-close to the optimal clustering. Specifically, we will show that
all are classified correctly except for the $6\epsilon n$ points in the bad case of Property 2 from Lemma \ref{lem:struct}.

Assume each cluster $C(x_j)$ contains
a majority of points that are $2t$ to a point in $G_j$ (we will prove this at the end).
Given a point $v\in C_j$ such that 
$d(v,c_i)-d(v,c_j)>\frac{\alpha w_{avg}}{2\epsilon}$ for all $c_i\neq c_j$ (Property 2 from Lemma \ref{lem:struct}),
and given a point $u\in C(x_j)$ that is at distance $2t$ to a point $u'\in G_j$, then
$d(v,u)\leq d(v,c_j)+d(c_j,u')+d(u',u)\leq d(v,c_j)+3t$.
On the other hand, given $u\in C(x_{j'})$ that is at distance $2t$ to a point $u'\in G_{j'}$, then
$d(v,u)\geq d(v,c_{j'})-d(c_{j'},u')-d(u',u)>18t+d(v,c_j)-3t\geq d(v,c_j)+15t$.
Then $v$'s median distance to $C(x_j)$ is $\leq d(v,c_j)+3t$, and $v$'s median distance to any other cluster is $\geq d(v,c_j)+15t$,
so $v$ will be assigned to the correct cluster.

Now we will prove each cluster $C(x_j)$ contains
a majority of points that are $2t$ to a point in $G_j$.
Assume for all $j$, $|C_j|>16|B|$. It follows that for all $j$ $|G_j|>15|B|$.
From the proof of Theorem \ref{thm:distr_as}, we know that $(\sum_j G_j\setminus (\sum_i C(v_j^i)))\leq 3|B|$,
therefore, for all $j$, $G_j^A>12|B|$, since $G_j^A$ represents the points in $A$ which are $2t$ to a point in $G_j$.
Again from the proof of Theorem \ref{thm:distr_as}, the clustering 
$\{G_1^A,\dots,G_k^A\}$ is $9|B|$-close to $G'=\{C(x_1),\dots, C(x_k)\}$.
Then even if $C(x_j)$ is missing $9|B|$ good points,
and contains $3|B|$ bad points, it will still have a majority of points that are within $2t$ of a point in $G_j$.
This completes the proof. 
\end{proof}

\section{Distributed Approximation Stability with Outliers} 
\label{sec:distr}

Next, we give a distributed algorithm for approximation stability with outliers using
$\tilde O\left(sk+z\right)$ communication.
However, as opposed to worst case, we can get close to the ground truth (target) clustering. In Section \ref{sec:lowerbounds}, we show a matching lower bound. 

\begin{theorem}(Distributed Clustering.) \label{thm:distr-aso}
Given a $(1+\alpha,\epsilon)$-approximation stable clustering instance,  Algorithm \ref{alg:distr-aso} runs in poly$\left(n^\frac{1}{\alpha}\right)$ time and with high probability outputs a clustering that is 
$O(\epsilon)$-close to $\OPT$ for $k$-median  
with $\tilde O\left(sk + z\right)$ communication
if each optimal cluster $\CC_i$ has cardinality at least $\max\left\{ 2\left(1+\frac{22}{\alpha}\right)\epsilon(n-z), \Omega\left(\frac{(n-z)}{sk} \right) \right\}$.
\end{theorem}

We start by giving intuition for our algorithm where there are no outliers.
The high-level structure of the algorithm can be thought of as a two-round version of the centralized
algorithm from approximation
stability with no outliers \cite{as}. Each machine effectively creates a coreset of its input, consisting of a weighted set of points,
and sends these weighted points to the coordinator.
The coordinator runs the same algorithm on these sets of weighted centers, to output the final solution.

In the analysis, we define good and bad points using Property $\emph{(1)}$ above with $y=20$ as opposed to $y=5$,
so that there are more bad points than in the non-distributed setting, 
$|B|=\left(1+\frac{1}{20}\right)\epsilon(n-z)$, but for each optimal cluster $\CC_i$,
the good points $G_i$ are even more tightly concentrated.
In the first round, each machine computes the neighborhood graph described above with parameter $\tau=\frac{w_{avg}}{10}$.
This more stringent definition of $\tau$ ensures that Claims \emph{(1)} and \emph{(2)} above are not only true for the input point set,
but also true for a summarized version of the point set, where each point represents a ball of data points within a radius of $\tau$.
Therefore, there is still enough structure present such that the coordinator can compute a near-optimal clustering,
and finally the coordinator sends the $k$ resulting (near optimal) centers to each machine.

Now we expand this approach to the case with outliers.
The starting point of the algorithm is the same: we perform two rounds of the sequential approximation stability algorithm with no outliers,
so that each machine computes a summary of its point set, and the coordinator clusters the points it receives.
Recall that in the centralized setting, running the non-outlier algorithm produces a list of clusters $\mathcal{X}$, some of which
are near-optimal and some of which are outlier clusters, and then we crucially computed the cost$_{\min}$ of each potential cluster
to distinguish the near-optimal clusters from the outlier clusters.
In the distributed setting, we can construct the set $\mathcal{X}$ using the two-round approach. 

However, the cost$_{\min}$ computation is sensitive to small sets of input points, and, as a result, the coresets will not give the 
coordinator enough information to perform this step correctly. 
In particular, this involves 
finding the closest points to a component that increase the cardinality to $\CC_{\min}$, and 
these points may be arbitrarily partitioned across the machines.

\begin{Frame}[\textbf{Algorithm \ref{alg:distr-aso}} : Distributed $k$-median with $z$-outliers under Approximation Stability]
\label{alg:distr-aso}
\ttx{Input}: Clustering instance $(V,d)$, cost $w_{avg}$, value $C_{min}$
\begin{enumerate}
    \item For each machine $i$, run Algorithm \ref{alg:neighborhood} with parameters 
$\tau=\frac{\alpha w_{avg}}{20\epsilon}$ and $b=\frac{\epsilon(n-z)}{s}$.
For each component $Q$ output of size $\geq\frac{\epsilon(n-z)}{s}$, choose an arbitrary point $c\in Q$ and send $(c,|Q|)$ to the coordinator.
    \item Given the set of weighted points received, $A$, run Algorithm \ref{alg:neighborhood} with parameters $\tau'=3\tau$ and 
$b=C_{min}-\left(1+\frac{22}{\alpha}\right)\epsilon(n-z)$. \label{step:aso-gprime}
    \item Label the components output of size $\geq b$ by $Q_1,\dots,Q_d$ and define $\mathcal{X}=\{Q_1,\dots,Q_d\}$.
    \item For each component $Q_i$, approximate $\text{cost}_{min}$ as follows: \label{step:dfor}
        \begin{enumerate}
        \item Sample $10\log n$ points uniformly at random from $Q_i$:
        the coordinator picks each point $(c,w_c)$ with probability proportional
        to its weight. The coordinator sends a request $(c,w_c)$ to the machine
        containing $c$, which then samples a point at random from $c$'s local component,
        sending this point to the coordinator. \label{step:sample}
        \item For each sampled point $c'$, compute $\min t$ such that
        $|B_t(c')|>$ $\max(C_{min},|Q_i|)$ over $V$, using binary search as follows.
        For each guess of $t$, send $(c',t)$ to each machine, and each
        machine returns $|B_t(c')|$ over its local dataset.
        \item For each $(c',t)$ pair computed in the previous step, compute
        $\text{cost}_{min}(c') := $
        
        $\sum_{v\in B_t(c')}d(c',v)$
        by having each machine send $\sum_{v\in B_t(c')\cap V_i}d(c',v)$.
        \end{enumerate}
    \item Create a new set 
$\mathcal{X}'=\{Q_i\mid\text{cost}_{min}(Q_i)<\left(1+\frac{11\alpha}{2}\right)\frac{1}{x}\cdot\OPT$. \label{step:distr-prune}
    \item For all $0\leq t\leq x$, for each size $t$ subset $\mathcal{X}'_t\subseteq\mathcal{X}'$ 
and size $\left(k-|\mathcal{X}'|-t\right)$ subset $\mathcal{X}_t\subseteq\left(\mathcal{X}\setminus\mathcal{X}'\right)$, 
        \begin{enumerate}
        \item Create a new clustering $\mathcal{C}=\mathcal{X}'\cup\mathcal{X}_t\setminus\mathcal{X}'_t$.
        \item For each cluster in $\mathcal{C}$, draw $10\log n$ random points using step \ref{step:sample} above.
        \item For each point $v\in V$, define $I(v)$ as the index of the cluster in $\mathcal{C}$ with minimum median distance 
        from the $10\log n$ points to $v$.
        \item Let $V'\subseteq V$ denote the $n-z$ points with the smallest values of $d(v,c_{I(v)})$,
        each center is restricted to the $10\log n$ random points. 
        For all $i$, set $Q_i'=\{v\in V'\mid I(v)=i\}$.
        \item If $\sum_i\text{cost}(Q_i')\leq (1+\alpha)\OPT$, return $\{Q_1,\dots,Q_k\}$.
        \end{enumerate}
\end{enumerate}
\ttx{Output:}  Connected components of $G'$ 
\end{Frame}

Furthermore, the centralized algorithm can try all possible centers to compute 
the minimum cost of a given component $Q$,
but in the distributed setting, to even find a point whose cost is a constant multiple of the 
minimum cost,
the coordinator needs to simulate random draws from $Q$ by communicating with each machine.
Even with a center $c$ chosen, the coordinator needs a near-exact estimate of the minimum 
cost of $Q$, however, it does not know the $\CC_{\min}$ closest points to $c$. 
To overcome these obstacles, our distributed algorithm balances accuracy with communication.

For each component $Q$, the coordinator simulates $\log n$ random draws from $Q$ by querying 
its own weighted points, and then querying the machine of the corresponding point. 
This allows the coordinator to find a center $c$ whose cost is only a constant factor away from 
the best center.
To compute $\text{cost}_{\min}(c)$, the coordinator runs a binary-search procedure with all 
machines to find the minimum distance $t$ such that $B_t(c)$ contains more than $\CC_{\min}$ 
points.

Given a random point $v$ from $Q$, by a Markov inequality, there is a $1/2$ chance that the cost of center $v$ on $V_c$ is at most
twice the cost with center $c$. From a Chernoff bound, by sampling $10\log n$ points for each component, 
each component will find a good center with high probability.
Therefore, the coordinator can evaluate the cost of each component up to a factor of 2,
which is sufficient to (nearly) distinguish the outlier clusters from the near-optimal clusters.
The rest of the algorithm is similar to the centralized setting.
We brute-force all combinations of removing $x$ low-cost clusters from $\mathcal{X}$ and adding
back $x$ high-cost clusters from $x$. We perform one more cluster purifying step, and then
check the cost of the resulting clustering.
If the cost is smaller than $(1+\alpha)w_{avg}(n-z)$, then we return this clustering.

\begin{proof} [Proof of Theorem \ref{thm:distr-aso}]
First we consider the case when $w_{avg}$ and $C_{min}$ are known.
Given machine $i$, let $\{G_1^i,\dots,G_k^i\}$ denote the good clusters intersected with $V_i$.
Define good points and bad points as in the previous section: 
a point is bad if it is not in the bad case of Property 1 for $y=20$,
or Property 2, otherwise a point is good.
For each $i$, the set of good points in $C_i$ is denoted $X_i$.
Recall from Lemma \ref{lem:struct}
that in the original dataset $V$, for all $i$, the good point set $X_i$ forms a clique in $G_{\tau}$ with
no neighbors in common with any points from different cores, and has at most $\epsilon (n-z)$ neighbors which are outliers. Here, $\tau=\frac{\alpha w_{avg}}{20\epsilon}$.
Therefore, if $|G_j^i|\geq\frac{\epsilon(n-z)}{s}$, it forms a component in $G_j'$ which does not contain core points from any other cluster,
and the total number of outliers added to a core component over all $j$, $i$, is less than $2\epsilon(n-z)$.
If $|G_j^i|<\frac{\epsilon(n-z)}{s}$, the component may be too small to have a point sampled and sent to the coordinator.
Over all machines, the total number of `missed' points from $X_j$ is at most $(s-1)\frac{\epsilon(n-z)}{s}\leq \epsilon(n-z)$.

Now we partition $A$ into sets $G_1^A,\dots, G_k^A,Z^A$,
where $G_j^A$ denotes points which are distance $2\tau$ to good points from $G_i$, and $Z'$ contains points which are far from all good points.
This partition is well-defined because any pair of good points from different clusters are far apart.
From the previous paragraph, for all $j$, the (weighted) size of $G_j^A$ is at least $|X_j|-\epsilon(n-z)\geq |C_j|-21\epsilon(n-z)$.
Again using Lemma \ref{lem:struct}, since each $u\in G_j^A$ was contained in a clique with a core point $u'$,
we have that for two points $u,v\in G_j^A$, there exist $u',v'\in G_j$ such that $$d(u,v)\leq d(u,u')+d(u',c_j)+d(c_j,v')+d(v',v)\leq 6\tau$$
Given $u\in G_j^A$ and $w\in G_{j'}^A$, there exist $u'\in G_j$, $w'\in G_{j'}$ such that $d(u',c_{j'})>18\tau-d(c_j,u')$, 
which we use to show $u$ and $w$ cannot have a common neighbor in $G_{3\tau}$.
Furthermore, at most $\epsilon(n-z)$ points in $Z^A$ can have a neighbor in $G_{3\tau}$ to a point in $G_j^A$, for al $j$.
It follows that for each $j$, $G'$ contains a component $G'_j$ containing $G_j^A$, such that $\{G_1',\dots,G_k'\}$ is
$22\epsilon(n-z)$-close to $\{G_1^A,\dots,G_k^A\}$. Since $|G_j^A|>C_{min}-21\epsilon(n-z)$, all of these components are added to $\mathcal{X}$.

Next, we show that just before step \ref{step:distr-prune}, $\mathcal{X}$ contains at most $x$ component outside of $\{G_1^A,\dots,G_k^A\}$.
From Lemma \ref{lem:badclique}, we know that at most $x$ outlier components of size $<C_{min}$ can have $\text{cost}_{min}$ cost smaller than 
$\left(3+\frac{2\alpha}{5}\right)\frac{1}{x}\OPT$.
The algorithm must determine an approximate $\text{cost}_{min}$ cost of each component in $\mathcal{X}$ whose size is $<C_{min}$, by communicating with each machine.
Given component $Q_i^A\in\mathcal{X}$ of size $<C_{min}$, let $Q_i$ denote the set of points `represented' by $Q_i^A$,
i.e., $Q_i=\{v\mid \exists a\in Q_i^A,j\text{ s.t. }v,a\in V_j\text{ and }d(v,a)\leq 2\tau\}$.
Let $q$ denote the optimal center for $Q_i$, and let $w_i$ denote the average distance $\frac{1}{|Q_i|}\sum_{v\in Q_i}d(q,v)$.
Let $c:=\text{argmin}_{c'}\sum_{v\in V_c}d(c,v)$ where $V_c$ denotes the $C_{min}$ closest points to $c$ subject to $Q_i\subseteq V_c$,
and let $Q'=V_c\setminus Q_i$.
By a Markov bound, at least half of the points $q'\in Q_i$ have $d(q,q')\leq 2w_i$.
Note that the algorithm is simulating $10\log d$ uniformly random draws from $Q_i$ in step  5 
By a Chernoff bound, at least one sampled point $\hat q$ must satisfy $d(q,\hat q)\leq 2w_i$ with high probability.
Then, 
\begin{align*}
\text{cost}_{min}(\hat q)&\leq\sum_{v\in Q_i}d(\hat q,v)+\sum_{v\in Q'}d(\hat q,v)\\
&\leq |Q_i|d(\hat q,q)+\sum_{v\in Q_i}d(q,v)+|Q'|d(\hat q,c)+\sum_{v\in Q'}d(c,v)\\
&\leq 2|Q_i|w_i+|Q_i|w_i+21\epsilon(n-z)(\frac{w_{avg}}{20\epsilon})+\sum_{v\in Q'}d(c,v)\\
&\leq 3\sum_{v\in Q_i}d(q,v)+\sum_{v\in Q'}d(c,v)+\frac{21}{20}\cdot w_{avg}(n-z)\\
&\leq 3\cdot\text{cost}_{min}(c)+\frac{21}{20}\cdot w_{avg}(n-z)
\end{align*}

Therefore, for all but $x$ good components $G_i^A$, the cost computed by the coordinator
will be $\leq 3\left(3+\frac{1\alpha}{20}\right)\frac{1}{x}\mathcal{OPT}$,
and all but $x$ bad components will have cost 
$> 3\left(3+\frac{1\alpha}{20}\right)\frac{1}{x}\mathcal{OPT}$.

Therefore, one iteration of step 6 will set $\mathcal{C}$ equal to
$\{G_1^A,\dots,G_k^A\}$, the near-optimal clustering.
As in the previous theorem, the final cluster purifying step will reduce the error
of the clustering down to cost $(1+\alpha)\OPT$, which must be $\epsilon$-close
to $\OPT$ by definition of approximation stability.

Now we move to the case where $w_{avg}$ and $C_{min}$ are not known.
For $w_{avg}$, we can use the same technique as in the previous sections: run an approximation algorithm for $k$-median with $z$-outliers to
obtain a constant approximation to $w_{avg}$.
For example, recently it was shown how to achieve an $7.08$-approximation in polynomial time \cite{krishnaswamy2017constant}.
Then we have a guess $\hat w$ for $w_{avg}$ that is in $[w_{avg},7.08w_{avg}]$. The situation is much like the case where $w_{avg}$ is known,
but the constant in the minimum allowed optimal cluster size increases by a factor of 7. The algorithm proceeds the same was as before.


Finally, we show how to binary search for the correct value of $C_{min}$.
If we run Algorithm \ref{alg:distr-aso} for $\hat C\in [22\epsilon(n-z),C_{min}]$, the number of edges in $G'$ in step 4 
must be a superset of the edges when $\hat C=C_{min}$.
However, since each core $X_i$ has fewer than $22\epsilon(n-z)$ neighbors outside of $X_i$, each core is still in a separate component of $G'$.
For each such component, $\text{cost}_{min}(C_i')$ still has cost 
$\leq 3\left(3+\frac{1\alpha}{20}\right)\frac{1}{x}\mathcal{OPT}$,
therefore, the number of good components with low cost after
step 7 is $\geq k-x$. 
If we run Algorithm \ref{alg:distr-aso} for $\hat C\in [C_{min},n]$, similar to the proof of Theorem \ref{thm:aso}, the number of components with cost 
$\geq 3\left(3+\frac{1\alpha}{20}\right)\frac{1}{x}\mathcal{OPT}$ after step 7 is $\leq k+x$ because there is at most one outlier component.
Therefore, the size of $\mathcal{X}$ as a function of $\hat C$ 
is monotone, and so we can perform binary search to find a value $\hat C$ such that
step 6 returns the optimal clustering.
%

The algorithm communicates $\tilde O(sk+z)$ bits to approximate $w_{avg}$.
The total communication in the first step is $O\left(sk\log n\right)$, since there are at most $\min\left\{\frac{s}{\epsilon}, O(sk) \right\}$ sets
of size at least $\max\left\{\frac{\epsilon n}{s}, \Omega(\frac{n}{sk}) \right\}$.
The communication for each component in step \ref{step:dfor} is $s\log (n)$, and $d\leq\frac{1}{\epsilon}$ (since each component is size $>\epsilon n$).
So the total communication in step \ref{step:dfor} is $\frac{s}{\epsilon}\log (n)$.
The binary search wrapper to find $C_{min}$ adds a $\log n$ multiplicative factor to the total communication.
Therefore, the total communication is $\tilde O\left(sk+z\right)$.
This completes the proof.
\end{proof}

\section{Communication Complexity Lower Bounds} \label{sec:lowerbounds}
In this section, we show lower bounds for the communication complexity
of distributed clustering with and without outliers.
We prove $\Omega(sk+z)$ lower bounds for two types of clustering problems:
computing a clustering whose cost is at most a $c$-approximation to the optimal
(or even just to determine the cost up to a factor of $c$) for any $c\geq 1$,
and computing a clustering which is $\delta$-close to $\OPT$, for any $\delta<\frac{1}{4}$.
This shows prior work is tight \cite{guha2017distributed}.

Our lower bounds hold even when the data satisfies a very strong, general notion of
stability, i.e. $c$-separation, for all $c \geq 1$. Recall, by Lemma \ref{lem:reductions}, an instance that satisfies $(\alpha n)$-separation satisfies almost all other notions of stability including approximation stability and perturbation resilience.
Furthermore, our lower bounds for $\delta$-close clustering hold even under a weaker
version of clustering, which we call \emph{locally-consistent clustering}. 
In this problem, instead of assigning a globally consistent index $[1,\dots,k]$ for each point, 
each player only needs to assign indices to its points that is consistent in a local manner, e.g.,
the assignment of indices $[1,\dots,k]$ to clusters $\{C_1,\dots,C_k\}$ chosen by player 1 
might be a permutation of the assignment chosen by player 2.

We work in the multi-party message passing model, 
where there are $s$ players, $P_1, P_2, \ldots, P_s$, who receive inputs $X^1$, $X^2$, \ldots $X^s$ respectively. They have access to private randomness 
as well as a common publicly shared random string $R$, and the objective 
is to communicate with a central coordinator who computes a function 
$f: X^1 \times X^2 \ldots \times X^s \to \{0, 1\} $ on the joint 
inputs of the players. The communication has multiple rounds and each player 
is allowed to send messages to the coordinator. Note, we can simulate communication 
between the players by blowing up the rounds by a factor of $2$. 
Given $X^i$ as an input to player $i$, let $\Pi\left(X^1, X^2, \ldots X^s\right)$ be 
the random variable that denotes the transcript between the players and the 
referee when they execute a protocol $\Pi$. For $i \in [s]$, let $\Pi_i$ 
denote the messages sent by $P_i$ to the referee. 

A protocol $\Pi$ is called a $\delta$-error protocol for function $f$ if there 
exists a function $\Pi_{out}$ such that for every input 
$Pr\left[\Pi_{out}\left(\Pi(X^1, X^2, \ldots X^s)\right) = f(X^1, X^2, \ldots X^s)\right] \geq 1 - \delta$. 
The communication cost of a protocol, denoted by $|\Pi|$, is the maximum length of
$\Pi\left(X^1, X^2, \ldots, X^s\right)$ over all possible inputs and random coin 
flips of all the $s$ players and the referee. The randomized communication 
complexity of a function $f$, $R_{\delta}(f)$, is the communication cost 
of the best $\delta$-error protocol for computing $f$.

\begin{definition}(\emph{Multi-party set disjointness} (\textsf{DISJ}$_{s,\ell}$).)
Given $s$ players, denoted by $P_1$, $P_2$, \ldots $P_s$, each player 
receives as input a bit vector $X^j$ of length $\ell$. 
Let $X$ denote the a binary matrix such that each $X^j$ is a column of $X$. 
Let $X_i$ denote the $i$-{th} row of $X$ and $X^j[i]$ denote the 
$(i,j)$-th entry of $X$. Then, \textsf{DISJ}$_{s,\ell} = \bigvee_{i\in [\ell]} \bigwedge_{j\in[s]}X^j[i]$, i.e. \textsf{DISJ}$_{s,\ell} =0$ if 
at least one row of $X$ corresponds to the all ones vector and $1$ otherwise. 
\end{definition}

We note that set disjointness is a fundamental problem in 
communication complexity and we use the following lower bound 
for \textsf{DISJ}$_{s,\ell}$ in the message-passing model by \cite{braverman2013tight}:

\begin{theorem}(Communication complexity of \textsf{DISJ}$_{s,\ell}$.)
\label{thm:adisjointness}
For any $\delta > 0$, $s = \Omega(\log(n))$ and $\ell \geq 1$,
the randomized communication complexity of multi-party set disjointness, $R_{\delta}(\textsf{DISJ}_{s,\ell})$, is $\Omega(s\ell)$.
\end{theorem}

Intuitively, we show a lower bound of $\Omega(sk)$ via a  
reduction from multi-party set disjointness with $s$ players, where each player get a bit vector of length $\ell=(k-1)/2$.
We first consider the case where $z=0$ and create a clustering instance as follows : 
we define upfront $2\ell+2$ possible locations
for the points: $\{ p_1,\dots, p_{\ell},q_1,\dots,q_{\ell},p,q \}$.
Now for all $i=1$ to $\ell$, each player creates a point at location $p_i$ if their input contains element $i$, i.e. the $i$-th coordinate of their bit vector is $1$, otherwise
they create a point at location $q_i$. The coordinator creates points at locations $p_1,\dots, p_\ell$ and $p$ and $q$. Note, the coordinator does not create any point at locations $q_1, \ldots q_{\ell}$.

If the set disjointness is a no instance, then there will be some element $i$ shared by all players. Observe, the number of unique locations in this case are $2\ell + 1 = k$, since every player inserts a point at location $p_i$ and no point is inserted at location $q_i$.
Therefore, it is easy to see that the optimal solution has cost $0$ since we can assign each unique location to its own cluster.
If the set disjointness is a yes instance, then there will be $2\ell + 2 = k+1$ distinct locations in the clustering
instance, so the optimal solution must have non-zero cost.
It follows that we can solve the original set disjointness instance by using a $c_1$-approximate clustering algorithm. We note that
the input can be made arbitrarily stable in either yes or no instances, by setting the
distances between $p_1,\dots,p_\ell,q_1,\dots,q_\ell$ arbitrarily far away from each other. We show a similar reduction works when the input instance has outliers. 


\begin{theorem} \label{thm:lb-apx}
Given $c_1\geq 1$,  the communication complexity for computing a $c_1$-approximation for
$k$-median, $k$-means, or $k$-center clustering is $\Omega(sk)$, even when promised that the
instance satisfies $c_2$-separability for any $c_2 \geq 1$.
Further, for the case of clustering with $z$ outliers, computing a $c_1$-approximation to $k$-median, $k$-means, or $k$-center cost, under the same promise requires $\Omega(sk+z)$ bits of communication.
\end{theorem}


\begin{proof}
The proof strategy we follow is to show that any  distributed clustering algorithm, $\mathcal{A}$, that achieves a $c_1$-approximation to $k$-median, $k$-means or $k$-center, given that the input satisfies $c_2$-separability, can be used to construct a distributed protocol, $\Pi$, that solves \textsf{DISJ}$_{s,\ell}$. Since the communication complexity of \textsf{DISJ}$_{s,\ell}$ is lower bounded by $\Omega(s\ell)$, this implies a lower bound on the communication cost of the distributed algorithm. 

First we consider the case when $z=0$. 
W.l.o.g. assume $k$ is odd, and set the length of each bit vector to be $\ell=(k-1)/2$.
We create a clustering instance as follows.
We define upfront the following set of $k+1$ locations on a graph. 
There is a clique of $2\ell$ locations such that all pairs of locations are distance $2\max(c_1,c_2, 1+\alpha)\textrm{poly}(n)$ apart.
Label these locations $p_1,\dots, p_{\ell},q_1,\dots,q_{\ell}$.
There are two additional locations, $p$ and $q$, such that the distance from $p$ to $q$ is 1, and the distance
from $p$ and $q$ to any other point is $2\max(c_1,c_2, 1+\alpha)\textrm{poly}(n)$.
Now for all $i=1$ to $\ell$, each player creates a point at location $p_i$ if it contains element $i$ (i.e., the $i$-th index of the bit vector is non-zero), otherwise
it creates a point at location $q_i$. The coordinator creates points at locations $p_1,\dots,p_\ell$ and $p$ and $q$.

If the set disjointness is a no instance,  i.e., \textsf{DISJ}$_{s,\ell} =0$, then there will be some element $i$ shared by all players.
Therefore, no point at location $q_i$ is ever created, there are $\leq k$ distinct locations in the clustering instance.
The optimal solution has cost $0$ by assigning each location to be its own cluster. We note 
that this clustering instance is $c_2$-separable, for any $c\geq 1$. 
To see this, note the maximum distance between two points in the same cluster is $0$ and the minimum distance across clusters is non-zero, therefore separability is satisfied for any $c_2 \geq 1$. 
Note, a 
similar argument for hard instances satisfying beyond-worse case assumptions carries through 
in subsequent reductions. 

If the set disjointness is a yes instance, i.e. \textsf{DISJ}$_{s,\ell} =1$ , then there will be $k+1$ distinct points in the clustering
instance. The optimal cost is 1 by putting $p$ and $q$ in the same cluster, and assigning all other points
to their own cluster. 
In this case and the previous case, the maximum distance between points in the same cluster (recall this distance is $1$),
is smaller than the minimum distance between points in different clusters (recall this is at least $2\max(c_1,c_2)\textrm{poly}(n)$) by a factor of at least $c_2\textrm{poly}(n)$ ,therefore the instance is $c_2$-separated for any $c_2 \geq 1$.  
It follows that we can solve the \textsf{DISJ}$_{s,\ell}$ instance by using a $c_1$-approximate clustering algorithm.
If the coordinator sees all of its points are given distinct labels, then it returns no. Otherwise, the coordinator returns yes. 
Recalling \textsf{DISJ}$_{s,\ell}$ has communication complexity $\Omega(sk)$ completes the lower bound.

Now we consider the case when $z>0$.
Note that we can assume $z\in \omega(sk)$ based on the previous paragraphs. Given $z\in \omega(sk)$
our goal is now to find an $\Omega(z)$ lower bound.
We construct a new clustering instance similar to the previous construction, but $s=2$ and $k=3$
(if $s>2$ or $k>3$, then only give nonempty input to the first two machines, or add $k-3$ points arbitrarily far away).
Then we give a reduction from 2-player set disjointness, \textsf{DISJ}$_{2,\ell}$, where Players 1 and 2 are given bit vectors $X^1,X^2\in\{0,1\}^\ell$,  as the input and
and the number of nonzero elements in each of $X^1$ and $X^2$ is $\ell/4$.
This version of set disjointness has communication complexity $\Omega(\ell)$ \cite{razborov1992distributional}.
We set $\ell=2z+4$, and we create a clustering instance as follows.
There is a clique of $\ell$ locations $p_1,\dots, p_{\ell}$ such that all pairs of locations are distance $2\max(c_1,c_2)\textrm{poly}(n)$ apart.
We also add locations $p$ and $q$ such that $d(p,q)=1$, and $p$ and $q$ are $2\max(c_1,c_2)\textrm{poly}(n)$ from the other points.
Each player $j\in \{1,2\}$, adds a point at location $p_j$ if $X^j[i]=1$, otherwise do not add a point. The coordinator creates points at locations $p$ and $q$.

Note the number of points created is $z+4$.
Similar to the previous paragraph, if set disjointness is a no instance, i.e. 
\textsf{DISJ}$_{2,\ell} =0$, then there is some index $i$ such that two points are at location
$p_i$. This implies that the number of unique locations is $z+3$. The optimal $k$-median with 
$z$ outliers solution is to make $p_i$ a cluster center, make $p$ and $q$ to be independent 
clusters, and make the $z$ remaining points to be outliers, so the total cost is zero.
If the set disjointness is a yes instance, i.e. \textsf{DISJ}$_{2,\ell} =1$, then each player 
inserts $\ell$ points in unique locations, and the coordinator $p$ and $q$. Therefore, there 
are $z+4$ points in different locations, 
so the optimal clustering is to put $p$ and $q$ into the same cluster, pick arbitrary $p_{i}, 
p_{i'}$ to be independent clusters and label the rest of the $p_{j}$'s as outliers.
Both yes and no instances also are $c_2$-separable, for any $c_2 \geq 1$. Further, the yes case has $0$ cost and the no 
case cost $1$, and thus any $c_1$-approximation clustering algorithm can distinguish between 
the two cases.
This completes the proof.

\end{proof}

By Lemma \ref{lem:reductions}, $\Omega(sk + z)$ is also a lower bound for instances that are $(1+\alpha, \epsilon)$-approximation stabile or $(1+\alpha)$-perturbation resilient for any $\alpha, \epsilon > 0$. 
We note that thus far we have ruled out a distributed clustering algorithm that has
communication complexity less than $\Omega(sk+z)$ to output the exact clustering under strong stability assumptions.  
Next, we prove the same communication lower bound holds when the goal is to return a 
clustering that is $\frac{1}{4}$-close to optimal in hamming distance. Note, this holds even
when the algorithm outputs a $c$-approximate solution to the clustering cost. 
Intuitively, the proof is again a reduction from \textsf{DISJ}$_{s,\ell}$, similar to the
proof of Theorem \ref{thm:lb-apx}.
The main difference is that we add roughly $\frac{n}{2}$ copies each of points $p$ and $q$.
If set disjointness is a no instance, $p$ and $q$ will each be in their own cluster,
but if it is a yes instance, then $p$ and $q$ must be combined into one cluster.
These two clusterings are $\frac{1}{2}$-far from each other, so returning a $\frac{1}{4}$-close solution requires solving set disjointness.

\begin{theorem} \label{thm:hamming}
Given $0<\delta<\frac{1}{4.01}$, the communication complexity for computing a clustering that is $\delta$-close to the optimal is $\Omega(sk + z)$, even when promised that the
instance satisfies $c$-separation, for any $c \geq 1$.
\end{theorem}

\begin{proof}
Again, the proof strategy we follow is to show that any  distributed clustering algorithm,
$\mathcal{A}$, that gets $\delta$-close to the optimal clustering,
given that the input satisfies $c$-separability, can be used to construct a distributed 
protocol, $\Pi$, that solves \textsf{DISJ}$_{s,\ell}$. Since the communication complexity of
\textsf{DISJ}$_{s,\ell}$ is lower bounded by $\Omega(s\ell)$, this implies a lower bound on
the communication cost of the distributed algorithm. 

Now we assume $k$ is even, and set the lengths of the input bit vectors to be $\ell=(k-2)/2$.
We create a clustering instance as follows:
we define locations $p_1,\dots,p_{\ell},q_1,\dots,q_{\ell}$ each distance 
$\max(c,1+ \alpha)\textrm{poly}(n)$ from each other.
There are two additional locations $p$ and $q$ distance $\max(c, 1+\alpha)\textrm{poly}(n)$ from the 
previous points, but $d(p,q)=1$.
Now for all $i=1$ to $\ell$, each player inserts a point at location $p_i$ if their input 
contains element $i$, i.e. if the $j$-th player has $X^j[i] =1$ they insert a point at $p_i$, else they insert a point at location $q_i$. The coordinator creates points at locations $p_1, p_2 \ldots p_{\ell}$.
Additionally, we also make $\frac{n}{2}-\ell\cdot(s+1)$ copies of both $p$ and $q$ and 
assign them to the coordinator. Next, we note that we set $n > \frac{4\ell(s+1)}{1 - 4\delta}$,
and rearranging the terms implies $\delta < \frac{1}{4} - \frac{(s+1)\ell}{2n}$. 

If the set disjointness is a no instance, i.e. \textsf{DISJ}$_{s,\ell} =0$, then there 
will be some element $i$ shared by all players. This implies all the players create points at location $p_i$. 
Therefore, no point at location $q_i$ is ever created, and there are $k$ distinct locations in the clustering instance.
The optimal solution has cost 0 by assigning each point to its own cluster. As seen before, this clustering instance is 
$c$-separable. 
Any $\delta$-close clustering must have all clusters of size at most $\frac{3n}{4}-\frac{3\ell\cdot (s+1)}{2}$, 
since $p$ and $q$'s clusters are both size $\frac{n}{2}-\ell\cdot (s+1)$ and $\delta n \leq \frac{n}{4} - \frac{(s+1)\cdot \ell}{2}$.

If the set disjointness is a yes instance, i.e. \textsf{DISJ}$_{s,\ell} =1$, 
then there will be $k+1$ distinct points in the clustering instance.
The optimal clustering is to put $p$ and $q$ in the same cluster, since all 
other points are arbitrarily far away. Then, the cluster that contains all 
copies of $p$ and $q$ together is of size $n - 2(s+1)\cdot\ell$. Observe, the maximum distance between 
points in every cluster is atleast $c$-factor smaller than the minimum distance between two points in different clusters. 
Therefore, this instance is $c$-separable, for all $c \geq 1$. Further, the optimal solution has cost $\Omega(n - 2(s+1)\cdot\ell)$. 
Recall, $\delta n < \frac{n}{4} - \frac{(s+1)\cdot\ell}{2}$.
Any $\delta$-close clustering must have one cluster of size 
$n - 2(s+1)\cdot\ell - \left(\frac{n}{4} - \frac{(s+1)\cdot\ell}{2}\right) > 
\frac{3n}{4} - \frac{3\ell\cdot (s+1)}{2}$, since the cluster
with $p$ and $q$ is size $n- 2\ell\cdot (s+1)$.


Observe, we have now reduced the problem to computing the cardinality of the largest cluster. 
The coordinator can determine the size of the largest cluster since he has access to all the 
copies of $p$ and $q$. Note, all other clusters are of size $1$. 
It follows that we can solve the original set disjointness instance by using a $\delta$-
Hamming distance clustering algorithm.
If the coordinator sees the largest cluster is size  greater than $ \frac{3n}{4}-\ell\cdot 
s$, then it returns yes. Otherwise, it returns no.
Since \textsf{DISJ}$_{s,\ell}$ has communication complexity $\Omega(s\ell)$, this implies an 
$\Omega(sk)$ lower bound for the case where $z=0$. 

Now we consider the case when $z>0$.
Note that we can assume $z\in \omega(sk)$ since we already know a $\Omega(sk)$ lower bound 
based on the previous paragraphs. 
Given $z\in \omega(sk)$
our goal is now to find an $\Omega(z)$ lower bound. We now only consider $2$-player 
disjointness, i.e. \textsf{DISJ}$_{2,\ell}$ and set $k =3$. Further, we are guaranteed that 
the inputs to the two players $X^1$ and $X^2$ have at most $\ell/4$ non-zero entries. Recall, 
the communication complexity of \textsf{DISJ}$_{2 ,\ell} = \Omega(\ell)$ 
\cite{razborov1992distributional}, therefore we set $\ell = 2z + 4$. Note, the two players 
receive as input length $\ell$ bit vectors $X^1$ and $X^2$, and construct a clustering 
instance as follows: first we define upfront the following set of $\ell+2$ locations on a 
graph. 
There is a clique of $\ell$ locations such that all pairs of locations are distance 
$\max(c_1,c_2)\textrm{poly}(n)$ apart.
Label these locations $p_1,\dots, p_{\ell}$.
There are two additional locations, $p$ and $q$, such that the distance from $p$ to $q$ is 
$1$, and the distance
from $p$ and $q$ to any other location is $\max(c_1,c_2)\textrm{poly}(n)$. 
For $i \in \{1,2\}$, for $j \in [\ell]$ if $X^j[i] =1$, player $j$ inserts a point at location $p_i$, 
else player $j$ does nothing. Additionally, the coordinator creates $\frac{n-z-1}{2}$ points 
at locations $p$ and $q$. Note, the total number of points created are $n$.  

If the set disjointness is a no instance, i.e. \textsf{DISJ}$_{2,\ell} =0$, then there will be some element
 $i$ shared by both the players. 
Therefore, two points lie at location $p_i$, and there are $\leq z + 3$ distinct locations in the set 
$\{ p_1,\dots, p_{\ell}, p, q\}$. Additionally, there are $\frac{n-z-2}{2}$ copies of points at $p$ and $q$.
The optimal solution assigns points at $p_i$, $p$ and $q$ to be their own clusters, and set the remaining 
$z$ distinct points to be outliers. Observe, the clustering cost of this solution is $0$.
Further, any $\delta$-close clustering must have all clusters of size smaller than 
$\frac{n-z-2}{2} + \delta(n-z) < \frac{n-z-2}{2} + \frac{n-z}{4} = \frac{3(n-z)}{4}-1$, 
since clusters at locations $p$ and $q$ are both size $\frac{n-z-2}{2}$ and $\delta (n-z) < \frac{n-z}{4}$.

If the set disjointness is a yes instance, i.e. \textsf{DISJ}$_{2,\ell} =1$, then there will 
be $z+4$ distinct points in the clustering instance. The optimal solution sets $z$ arbitrary 
points in the set $\{ p_1,\dots, p_{\ell}\}$ to be outliers. Two points in 
this set remain and they have to be assigned as their own clusters. W.l.o.g., let these 
points be $p_i$ and $q_j$. This forces $p$ and $q$ to be the same cluster. This clustering 
incurs cost $\frac{n -z -2}{2}$ and has cardinality $n-z-2$.  
Recall, $\delta n < \frac{n-z}{4}$. Therefore, any $\delta$-close clustering must have one 
cluster of size $n - z - 2 -  \delta(n-z) > n - z - 2 - \left(\frac{n-z}{4}\right) = 
\frac{3(n - z)}{4} - 2$, since the cluster
that contains all points at $p$ and $q$ is of size $n- z-2$. Since the clusters have integer cardinalities, having 
cluster size strictly greater than $\frac{3(n - z)}{4} - 2$ is equivalent to a cluster size 
of $\geq \frac{3(n - z)}{4} - 1$. Therefore, the largest cluster size in the two cases are 
disjoint. Since the coordinator has all the copies of points at $p$ and $q$, he can determine which 
case we are in and in turn solve \textsf{DISJ}$_{2,\ell}$, which completes the proof. 
\end{proof}

Though the above lower bounds are quite general, it is possible that the hard instances may 
have the optimal clusters to be very different in cardinality if $sk$ is large.
The smallest cluster may be size $O\left(\frac{n}{sk}\right)$, while the largest cluster may 
be size $\Omega(n)$. 
Often, real-world instances may have balanced clusters. 
Therefore, we extend our previous lower bounds to the setting where we are promised 
that the input clusters are well balanced, i.e. have roughly the same cardinality. We also consider algorithms that only get $\delta$-close to the optimal clustering. We are further promised that the input instance 
satisfies $(1+\alpha, \epsilon)$-approximation stability and show lower bounds in this setting. We note that the combination of 
these assumptions is really strong yet we can show non-trivial lower bounds in this setting, indicating that $\Omega(sk + z)$ communication is fundamental barrier in distributed clustering.
We begin by defining the following basic notions from information theory: 
\begin{definition}(Entropy and conditional entropy.)
The entropy of a random variable $X$ drawn from distribution $\mu$, denoted as $X \sim \mu$, 
with support $\chi$, is given by 
$$
H(X) = \sum_{x \in \chi} \Pr_{\mu}[X = x]\log\frac{1}{\Pr_{\mu}[X = x]}
$$
Given two random variable $X$ and $Y$ with joint distribution $\mu$, the entropy of $X$ conditioned on $Y$ is given by
$$
H(X\mid Y) =  \expecf{y\sim \mu(Y)}{\sum_{x \in \chi} \Pr_{\mu(X\mid Y =y)}[X = x]\log \frac{1}{\Pr_{\mu(X\mid Y =y)}[X = x]} }
$$
\end{definition}

Note, the binary entropy function $H_2(X)$ is the entropy 
function for the distribution $\mu(X)$ supported on 
$\{0,1\}$ such that $\mu(X)=1$ with probability $p$ and $\mu(X) =0$ otherwise. 

\begin{definition}(Mutual information and conditional mutual information.)
Given two random variables $X$ and $Y$, the mutual information between $X$ and
$Y$ is given by
$$
I(X;Y) = H(X) - H(X\mid Y) = H(Y) - H(Y\mid X)
$$
The conditional mutual information between $X$ and $Y$, 
conditioned on a random variable $Z$ is given by 
$$
I(X;Y\mid Z) = H(X\mid Z) - H(X\mid Y,Z) = H(Y\mid Z) - H(Y\mid X,Z)
$$
\end{definition}

\begin{definition}(Chain rule for mutual information.)
\label{eqn:chain_rule}
Given random variables $X_1, X_2, \ldots X_n$, $Y$ and $Z$, the chain rule for mutual information is defined as
$$
I(X_1, X_2, \ldots X_n; Y\mid Z) = \sum^{n}_{i=1} I(X_i; Y\mid X_1, X_2, \ldots X_{i-1}, Z )
$$
\end{definition}

Recall, the $\delta$-error randomized communication complexity of $\mathcal{A}$, $R_{\delta}
(\mathcal{A})$, in the message passing model is communication complexity of any randomized 
protocol $\Pi$ that solves $\mathcal{A}$ with error at most $\delta$. 
Let $X^1, X^2, \ldots X^s$ be the inputs for players $P_1, P_2, \ldots P_s$. 
Let $\mu$ be a distribution over $X^1, X^2, \ldots X^s$. 
We call a deterministic protocol $(\delta,\mu)$-error if it gives the correct answer for 
$\mathcal{A}$ on at least a $1 - \delta$ fraction of the input, 
weight by the distribution $\mu$. Let $D_{\mu, \delta}(\mathcal{A})$ denote the cost of the
minimum communication $(\delta, \mu)$-error protocol. By Yao's minimax lemma, we know that
$R_{\delta}(\mathcal{A}) \geq \textrm{max}_{\mu} D_{\mu, \delta}(\mathcal{A})$.
Therefore, in order to lower bound the randomized communication complexity of $\mathcal{A}$,
it suffices to construct a distribution $\mu$ over the input such that any deterministic 
protocol that is correct on $1-\delta$ fraction of any input can be analyzed easily.
We note that the communication complexity of a protocol $\Pi$ is further lower bounded by
it's information complexity. 

\begin{definition}(Information complexity of $\mathcal{A}$.)
\label{def:ainformation_complexity}
For $i \in [s]$, let $\Pi_i$ be a random variable that denotes the transcript of 
the messages sent by player $P_i$ to the coordinator. We overload notation by letting $\Pi$ 
denote the
concatenation of $\Pi_1$ to $\Pi_s$. Then, the information complexity of $\mathcal{A}$ is 
given by 
\[
\textsf{IC}_{\mu, \delta}(\mathcal{A}) = \min_{(\delta, \mu)-\textrm{error } \Pi} I (X_1, X_2, \ldots X_s ; \Pi)
\]
\end{definition}

By a theorem of \cite{huang2015communication}, we know that $R_{\delta}(\mathcal{A}) \geq 
\textsf{IC}_{\mu, \delta}(\mathcal{A})$.
Therefore, our proof strategy is to design a distribution $\mu$ over the input and lower 
bound the information complexity of the resulting problem. Critically, this relies on lower 
bounding the mutual information between the inputs for each player and the resulting protocol 
$\Pi$.

\begin{theorem} \label{thm:alice}
Given $\delta<\frac{1}{4}$ and the promise that the optimal clusters are balanced, i.e., the 
cardinality of each cluster is $\frac{n}{k}$, the communication complexity for computing a 
clustering that is $\delta$-close to the optimal $k$-means or $k$-median clustering
is $\Omega(sk)$. 
\end{theorem}

\begin{proof}
We begin with an $\Omega(k)$ lower bound for 2 players, and 
subsequently we will extend it to $s$ players. Denote player 1 by Alice, 
and player 2 by Bob. Let $X^1$ and $X^2$ denote the length $\ell$ 
bit vectors given as input to Alice and Bob respectively. 
We first describe the clustering instance that is created by 
Alice and Bob based on their input. Let $\ell = \frac{k}{2}$ 
and let $X^{j}[i]$ denote the $i$-th entry of the $j$-th bit vector. 
Consider 2-dimensional Euclidean space, $\mathbb{R}^2$. If $X^1[1] = 0$,
Alice constructs the points $\{(-3,1),(-3,-1)\}$ else she constructs 
the points $\{(0,1),(0,-1)\}$. If $X^2[1] = 0$, Bob constructs the 
points $\{(3,1),(3,-1)\}$, else he constructs the points $\{(0,1),(0,-1)\}$. 
Alice and Bob then repeat the above construction $k/2$ times, moving the 
gadgets arbitrarily far away from each other to ensure that no two
points from different gadgets get put into the same cluster.  

Focusing on the first gadget, we observe that if Alice and Bob both have $X^{1}[1]=X^{2}
[1]=1$, the point set $\{(0,1),(0,-1)\}$, the optimal $2$-clustering cost is $0$. In any
other case, the optimal clustering is for Alice's two input points to be a single cluster 
and Bob's two input points to be a single cluster. The same is true for Bob. Both Alice 
and Bob are aware of this setup, so the only unknown for Alice is a single bit representing
which of the two input pairs Bob received, i.e. $X^{2}[1]$. Similarly, the only unknown for
Bob is a single bit, $X^{1}[1]$.

In total, there are $2k$ input points, and $\mathcal{OPT}$ is composed of a union of the
$k/2$ optimal 2-clusterings, one from each gadget. Recall, 
$R_{\delta}(\mathcal{A}) \geq  \textsf{IC}_{\mu, \delta}(\mathcal{A})$, 
therefore we define a distribution $\mu$ over the input as follows: 
Each entry of $X^1$ and $X^2$ is $1$ with probability $1/2$ and $0$ otherwise. 
Recall, a $(\delta, \mu)$-error protocol $\Pi$ achieves the correct answer on at least 
a $1-\delta$ fraction of the input, i.e. it gets at least $1-\delta$ gadgets right. 
Further, we observe that if a clustering $\mathcal{C}$ is $\delta$-close to $\mathcal{OPT}$,
then it solves a $1-2\delta$ fraction of the $2$-clustering gadgets. Therefore, 
a distributed clustering algorithm that gets $\delta$-close to $\mathcal{OPT}$ 
achieves a $(2\delta, \mu)$-protocol. It remains to show that can lower bound
$\textsf{IC}_{\mu, 2\delta}$ for such a $\mu$. 
From definition \ref{def:ainformation_complexity}, it follows that
\begin{equation*}
\begin{split}
\textsf{IC}_{\mu, 2\delta}(\mathcal{A})  = I(X_1, X_2; \Pi) 
& = I(X_1;\Pi \mid X_2) + I(X_2 ;\Pi\mid X_1) \\
& \geq I(X_1;\Pi \mid X_2 ) \\
& \geq \Omega(\ell) = \Omega(k) 
\end{split}
\end{equation*}
where the first equality follows from the definition of information complexity, 
the second follows from the chain rule of mutual information (definition \ref{eqn:chain_rule}), 
the third follows from mutual information being non-negative and the last follows from Alice learning at least a $1-\delta$ fraction of Bob's input for which $X^1 = X^2 = 1$. 
Therefore, $R_{\delta}(\mathcal{A}) = \Omega(k)$, which completes the proof for $2$ players.

 Now we extend the construction to $s$ players to achieve an $\Omega(sk)$ bound.
WLOG, assume that $s$ is even. Create inputs for $s/2$ players equal to the inputs Alice, and set the inputs for the remaining $s/2$ players equal to the input for Bob.
Specifically, the $s/2$ players that mimic Alice all receive the same input $X$, and the $s/2$ players that mimic Bob receive the same input $Y$.
Then $\mathcal{OPT}$ is the same as in the two-player case, but with each point copied $s/2$ times. Observe, if a clustering $\mathcal{C}$
is $\delta$-close to $\mathcal{OPT}$, for $\delta < 1/4$, then at least half of the players mimicking ``Alice'' learn the solution to at least a
$1-\Theta(\delta)$ fraction of the gadgets. Recall, from the previous paragraph, Alice requires $\Omega(k)$ bits to learn a $1-\delta$ fraction of the clustering. In order to communicate this to $\Omega(s)$ other places, the total communication is $\Omega(sk)$, which implies the overall $\Omega(sk)$ lower bound.
Note there are only $\Theta(k)$ bits needed to specify the input for every player, since there are only two distinct inputs
each given to half the players. However, we are still able to obtain the $\Omega(sk)$ lower 
bound since this information needs to
travel to $\Omega(s)$ different players so that all players can output a correct clustering. 

\end{proof}

Next, we extend the above lower bound to clustering instances that are balanced and also satisfy 
$(1+\alpha, \epsilon)$-approximation stability. Perhaps surprisingly, we show that there is 
no trade-off between the stability parameters and the communication lower bound even if the 
clusters are balanced and the algorithm outputs a clustering that is $\delta < \epsilon/4$ 
close to the optimal clustering. In contrast, our previous result can handle all $\delta < 
1/4$. We begin by introducing a promise version of the multi-party set disjointness problem, 
where the promise states if the sets intersect, the intersect on exactly one element. 
Formally,

\begin{definition}(Promise multi-party set disjointness (\textsf{PDISJ}$_{s,\ell}$).)
Given $s$ players, denoted by $P_1$, $P_2$, \ldots $P_s$, each player 
receives as input a bit vector $X^j$ of length $\ell$. 
Let $X$ denote the a binary matrix such that each $X^j$ is a column of $X$. 
Let $X_i$ denote the $i$-{th} row of $X$ and $X^j_{i}$ denote the 
$(i,j)$-th entry of $X$. We are promised that at most one row of $X$ has all ones. Then, 
\textsf{PDISJ}$_{s,\ell} = \bigvee_{i\in [\ell]} \bigwedge_{j\in[s]}X^j_{i}$, i.e. 
\textsf{PDISJ}$_{s,\ell} =0$ if any row of $X$ corresponds to the all ones vector and $1$ 
otherwise. 
\end{definition}

We use a result of \cite{bar2004information} to lower bound the communication complexity of 
set-disjointness in the multi-party communication model. 

\begin{theorem}(Communication complexity of \textsf{PDISJ}$_{s,\ell}$ 
\cite{bar2004information}.)
\label{thm:apromise_disjointness}
For any $\delta > 0$, $s, \ell \in \N$,
the randomized communication complexity of promise multi-party set disjointness, $R_{\delta}
(\textsf{PDISJ}_{s,\ell})$, is $\Omega(\ell/s^2)$.
\end{theorem}

We show that an algorithm obtaining a $\delta$-close clustering, given the clusters are 
balanced and the clustering instance is $(1+\alpha, \epsilon)$-stable can be converted into a 
randomized communication protocol that solves \textsf{PDISJ}$_{s,\ell}$. 

\begin{theorem}
\label{thm:final_lowerbound}
Given a $(1+\alpha,\epsilon)$-approximation stable instance with $z$ outliers
such that $\epsilon = o(1)$ and $\delta< \frac{\epsilon}{4}$, and the promise that the 
optimal clusters are balanced, i.e.,\ the 
cardinality of each cluster is $\frac{n-z}{k}$, the communication complexity for computing a 
clustering that is $\delta$-close to the optimal $k$-means or $k$-median clustering
is $\Omega(sk+z)$.
\end{theorem}


\begin{proof}
We extend the previous proof to show the lower bound still holds if the input clustering 
instance satisfies 
approximation stability. 
Given $\delta<\frac{\epsilon}{4}<\frac{1}{4}$, first we show
that to achieve any $(1+\alpha)$-approximation to the optimal cost,
we cannot output a cluster containing points from different gadgets. Then, we introduce a 
communication problem that is a variant of set-disjointness and show that any clustering 
algorithm that gets $\delta$-close to an optimal clustering must indeed solve set-
disjointness with good probability.  We then invoke the set disjointness lower-bound from 
Theorem \ref{thm:apromise_disjointness}. 

Recall, Alice and Bob receive length $\ell$ bit vectors $X^1$ and $X^2$ as input. Let $\ell = 
\frac{k}{2}$ 
and let $X^{j}[i]$ denote the $i$-th entry of the $j$-th bit vector. 
Instead of constructing points $\{(-3,1),(-3,-1)\}$ or 
$\{(0,1),(0,-1)\}$, Alice now constructs two points at $(-L,0)$ if $X^{1}[0] = 0$ or 
constructs $\{(0,1),(0,-1)\}$ if $X^{1}[0] = 1$.
Similarly, Bob now constructs two points at $(L,0)$ if $X^{2}[0] = 0$ or $\{(0,1),(0,-1)\}$ 
otherwise. 

In total, there are $2k$ input points, and $\mathcal{OPT}$ is composed of a union of the
$k/2$ optimal 2-clusterings, one from each gadget. 
By setting $L > 10 (1+\alpha) \mathcal{OPT}$, it is easy to see that clusters within the same 
gadget that have unique $x$-coordinates cannot swap points and still obtain a $(1+\alpha)$-
approximation to the optimal cost. Therefore, 
the only possible clusters in a 
$(1+\alpha)$-approximate clustering that swap points must share their $x$-coordinate. 
Alice and Bob then repeat the above construction $k/2$ times, moving the 
gadgets arbitrarily far away from each other to ensure that no two
points from different gadgets get put into the same cluster while maintaining a $(1+\alpha)$-
approximation to the clustering cost. We fist show a sufficient condition under which the 
above construction is $(1+\alpha, \epsilon)$-stable clustering instance. Then, we show that 
any algorithm that gets $\delta$-close to the optimal clustering must communicate 
$\Omega(sk)$ bits.

Focusing on the first gadget, we observe that if Alice and Bob both have $X^{1}[1]=X^{2}
[1]=1$, the point set is $\{(0,1),(0,-1)\}$, and the optimal $2$-clustering cost is $0$. Alice's two points lie in different clusters and Bob is symmetric. In any
other case, the optimal clustering is for Alice's two input points to be a single cluster 
and Bob's two input points to be a single cluster. In the case where the input for Alice is $0$, the clustering is determined and the cost is $0$. The same holds for Bob. Therefore, the only case in which the clustering instance has non-zero cost is when the input on the first index is $(0,1)$ or $(1,0)$. In such as case, the clustering cost is $4$. 
Both Alice 
and Bob are aware of this setup, so the only unknown for Alice is a single bit representing
which of the two input pairs Bob received, i.e. $X^{2}[1]$. Similarly, the only unknown for
Bob is a single bit, $X^{1}[1]$. In every case, each cluster has cardinality $2$, and therefore the instance is balanced. 

Next, if the number of coordinates $i$ such that $X^{1}[i]= X^{2}[i]= 1$ is at most $\epsilon 
k$, we observe that the instance is $(1+\alpha, \epsilon)$-stable. To see this, observe that 
any $(1+\alpha)$-approximation to the cost can change only swap points when the two optimal 
clusters for a given gadget share the same $x$-coordinate. Note, in all other cases, the 
clusters are at least $L$ apart, and the cost cannot be a $(1+\alpha)$-approximation. The 
optimal clusters share the same $x$-coordinate only when $X^{1}[i]= X^{2}[i]= 1$ and if the 
points switch from their optimal cluster, the cost increases by $2$ units. However, since 
there are at most $\epsilon k$ such gadgets overall, at most $8\epsilon k = 4\epsilon n$ 
points can switch from their optimal clusters without blowing the cost more than a 
$(1+\alpha)$-factor. Therefore, rescaling $\epsilon$ by $4$, the instance is $(1+\alpha, 
\epsilon)$-stable. 

Finally, we describe how a clustering algorithm that obtains a $\delta$-close approximation 
to the optimal clustering in the aforementioned instance is a valid protocol for solving 
\textsf{PDISJ}$_{s,\ell}$. Given an instance of \textsf{PDISJ}$_{2,k/2}$, Alice and Bob 
create $\epsilon n - 1= 2\epsilon k - 1$ dummy indices that are set to $1$ for both Alice and 
Bob. Note, given the promise, this Alice and Bob have at most $2\epsilon k$ coordinates that 
are $(1,1)$ and as discussed previously, the resulting clustering instance is $(1+\alpha, 
\epsilon)$-stable. Let the $\ell' = k/2 + 2\epsilon k - 1$. Using public randomness, Alice 
and Bob agree on a uniformly random permutation $\pi: [\ell'] \to [\ell']$. They now randomly 
permute their input along with the dummy coordinates and run the clustering protocol. 

Observe, since the clustering protocol outputs a $\delta$-close solution for $\delta < 
\frac{\epsilon}{4}$ at least $1 - 2\delta \geq 1-\epsilon/2$ fraction of the points get 
classified correctly. Further, each cluster has cardinality $2$, therefore at least $(1 - 
\epsilon)$-fraction of the clusters would be the optimal clusters. Since we uniformly permute 
the indices of the input before running the protocol, for any given index, the corresponding 
cluster has hamming distance $0$ from the optimal clustering with probability at least $1 - 
\epsilon$. In other words at most $\epsilon$-fraction of the clusters are incorrect. The 
protocol outputs a clustering that is known to both Alice and Bob. For each index of their 
input, they know whether their pair of points lie in the same cluster of different clusters. 
Let $\mathcal{I}$ be the set of indices for which Alice and Bob's points lie in different 
clusters. If  $\mathcal{I} > 4 \epsilon k$, protocol outputs fail. Else, Alice communicates 
her input on the set $\mathcal{I}$ to Bob. Bob applies $\pi^{-1}$ to $\mathcal{I}$, and 
verifies if the indices correspond to the dummy indices that were added or indeed the sets 
are not disjoint. Note the verification step requires additional communication. Since 
$\mathcal{I} \leq 4 \epsilon k$, and $\epsilon = o(1)$, the total additional communication is 
$o(k)$.    

Consider the case where the sets are not disjoint. Then, there is an index $i^*$ such that 
the input $X^1[i^*]=X^2[i^*]=1$ and  
with probability at least $1 -\epsilon$, the clustering algorithm (protocol) correctly 
clusters the corresponding $2$-means gadget. This implies that Alice and Bob know that their 
pair of points lie in different clusters, thus $i^* $ is in the set $\mathcal{I}$ and Alice 
communicates $X^1[i^*]$ to Bob. Bob can then verify that $\pi^{-1}(i^*)$ is not a dummy index 
and that $X^1[i^*] = X^2[i^*] = 1$. 
The case where the sets are disjoint is more subtle. In this case, the clustering algorithm 
may return $4\epsilon k$ indices such that Alice's points belong to separate clusters, i.e. 
they correspond to a $(1,1)$ input, therefore leading to false positives. However, we observe 
that we can verify if the sets are disjoint by Alice sending over her input bits on the set 
$\mathcal{I}$ to Bob. Bob can verify if they correspond to the dummy indices and the sets are 
indeed disjoint. Note, this increases the over all communication by $o(k)$. We note that by 
Theorem \ref{thm:apromise_disjointness}, the communication of the protocol is $\Omega(k - 
\epsilon k) = \Omega(k)$. We then use the previous technique of cloning the Alice and Bob 
players $s/2$ times each, therefore, communicating the solution to each player requires 
$\Omega(sk)$ bits of communication. 

Next, we extend out lower bound to the case where the input has $z$ outliers, the clusters 
are balanced and the instance satisfies $(1+\alpha, \epsilon)$-stability with outliers. In 
order to show a communication lower bound we consider the \textsf{PDISJ}$_{2,\ell}$ problem 
and show that a protocol that solves an instance that satisfies the above assumptions in turn 
solves \textsf{PDISJ}$_{2,\ell}$ with good probability. 
Alice and Bob receive length $\ell$ bit vectors $X_1$ and $X^2$ as input. Let $\ell = 
\frac{2z+4}{2}$ 
and let $X^{j}[i]$ denote the $i$-th entry of the $j$-th bit vector. Alice and Bob also pad 
their input to be length $2\ell$ with additional $1$s as follows: indices $\ell$ to 
$\ell+\ell/4$ are reserved for Alice and indices $\ell + \ell/4$ to $\ell + \ell/2$ are 
reserved for Bob. Alice counts the number of $1$s she receives from the input to the 
\textsf{PDISJ}$_{2,\ell}$ instance, and pads $1$s in the indices allocated to her to make the 
total number of $1$s be $\ell/4$ and sets the remaining indices to $0$. Similarly, Bob pads 
$1$s in the indices allocated to him to make the total number of $1$s be $\ell/4$ and sets 
the remaining indices to $0$. 

Note, we now have an instance such that Alice and Bob have 
exactly $\ell/4$ non-zero entries in their bit vectors and at most one index contributes to 
the intersection. Further, Alice and Bob append $\frac{k-1}{2}$ coordinates and both set 
their bit vectors corresponding to these instances to $1$. Let $X^{1'}$ and $X^{2'}$ denote 
the padded bit vectors for Alice and Bob and let $\ell' = \ell + \ell/2 + \frac{k-1}{2}$.  
Observe, the total number of indices that correspond to $X^{1'}[i]= X^{2'}[i] = 1$ is at most 
$k$ and at least $k-1$, toggled by \textsf{PDISJ}$_{2,\ell}$ being $1$ or $0$. Using public 
randomness, Alice and Bob agree on a uniformly random permutation $\pi: [\ell'] \to [\ell']$. 
They now apply the permutation $\pi$ to the bit vectors $X^{1'}$ and $X^{2'}$ locally and 
create a clustering instance. 

Upfront, we set locations $p_1, p_2, \ldots p_{\ell'}$, such that the pair-wise distance 
between these locations is $\textrm{max}(c_1, c_2)\textrm{poly}(n)$. Each player creates a 
point at location $p_i$ if $X^{j}[i]=1$, else they do not create any points. The players then 
execute the distributed clustering algorithm on this instance. The clustering algorithm 
allows Alice and Bob to figure our which of their indices correspond to optimal clusters and 
which indices correspond to outliers. Let $\mathcal{I}$ be the set of indices that correspond 
to clusters for Alice. Alice then communicates her input on the index set $\mathcal{I}$ to 
Bob. Bob applies the inverse permutation, $\pi^{-1}$, to $\mathcal{I}$ and verifies if there 
is some index $i^*$ corresponding to the \textsf{PDISJ}$_{2,\ell}$ instance such that $X^{1'}
[i^*]= X^{2'}[i^*] = 1$.

It is left to argue that the above protocol creates clustering instances that are $(1+\alpha, 
\epsilon)$-stable and balanced, any $\delta$-close clustering algorithm indeed solves the 
original \textsf{PDISJ}$_{2,\ell}$ instance and the communication overhead in sending the 
bits corresponding to the index set $\mathcal{I}$ is small. We note here that the clustering 
instance is stable only when \textsf{PDISJ}$_{2,\ell} = 0$, i.e. the clustering algorithm 
outputs a $\delta$-close solution only in the above case. If the algorithm is provided an 
instance that is not stable, it is allowed to fail. If \textsf{PDISJ}$_{2,\ell} = 0$, there 
exists some index $i^* \in [\ell]$ such that $X^{1'}[i^*]= X^{2'}[i^*] = 1$. Therefore, there 
exists a location $p_{\pi(i^*)}$ such that both Alice and Bob created points at this 
location. There are $k-1$ additional locations corresponding to the dummy indices that have 
two points. 

Note, the optimal clustering is then to create a center at every location that 
contains two points and label the remaining $z$ points as outliers. Therefore, the optimal 
clustering cost is $0$. Similar to the $z=0$ proof, a $\delta$-close 
protocol must output $1-4\delta$ clusters correctly. Therefore, with probability at least $1-
\delta$, the clustering algorithm outputs the cluster at location $p_{\pi(i^*)}$ correctly. 
Note, the total number of clusters output should be at most $k$, in turn upper bounding the 
cardinality of $\mathcal{I}$. It is easy to see that Bob can then verify if 
\textsf{PDISJ}$_{2,\ell} = 0$ with $O(k)$ extra communication. Combining this with Theorem 
\ref{thm:apromise_disjointness}, the communication complexity of the protocol is $\Omega(\ell 
- k)$, and since $\ell = \frac{2z+4}{2}$ and $k = o(z)$, the overall communication is 
$\Omega(z)$, which completes the proof. 
\end{proof}


\newpage
\bibliography{clustering}

\newpage

\appendix
\section{Beyond the $\Omega(sk+z)$ Lower Bound}

In some clustering settings, a full assignment of every datapoint to a cluster index might not
be necessary. For instance, we may only need to know the mean of the optimal clusters,
or we may only need to compute cluster assignments online as queries come in.
Now we present an algorithm that uses much less communication to handle
these cases. Specifically, the algorithm uses $O(s\log n+\frac{1}{\epsilon}\log n)$
communication and outputs a function $f$ which can be used to cluster all input points
(but the size of the cluster is too large to send to each machine, which would lead to a full
clustering).
The algorithm is based on subsampling the clustering instance,
inspired by Balcan et al.~\cite{balcan2009agnostic}.

We present an algorithm that uses $O(s\log n+\frac{1}{\epsilon}\log n)$ communication, and clusters
a \emph{sample} of the input points, and then creates a function $f$ which can be used to cluster all input points
(but sending the function to each machine would require $\Theta(sk)$ communication).
It is still an open question whether it is possible to fully cluster all input points with $o(sk)$ communication.
Formally, the theorem is as follows.

\begin{theorem} \label{thm:distr_sample}
Algorithm \ref{alg:distr_sample} takes as input a clustering instance satisfying $(1+\alpha,\epsilon)$-approximation stability
such that each optimal cluster is size at least $(6+\frac{30}{\alpha})\epsilon n+2$
and outputs a function $f:V\rightarrow [k]$ defining a clustering that is $\epsilon$-close to $\mathcal{OPT}$.
The communication complexity is $O(s\log n+\frac{1}{\epsilon}\log n)$.
\end{theorem}

Note that we can cluster any subset $S\subseteq V$ of points in time $O(|S|)$ by sending $S$ to the coordinator and using
$f$ to cluster $S$. But if the goal is to cluster every single point $V$, then we need to use $\Theta(sk)$ communication.

\begin{Frame}[\textbf{Algorithm \ref{alg:distr_sample}} : Distributed $k$-median Clustering under $(1+\alpha,\epsilon)$-approximation stability for large Clusters]
\label{alg:distr_sample}
\ttx{Input}: Distributed points $V=V_1\cup\cdots\cup V_s$
\begin{enumerate}
    \item  Each machine $i$ sends $|V_i|$ to the coordinator
    \item For all $i$, the coordinator computes $s_i=\frac{n_i}{n}\cdot\frac{1}{\epsilon}\log 10k$  and sends $s_i$ to machine $i$.
    \item Each machine $i$ selects $s_i$ points at random and sends them to the coordinator. The coordinator collects all sampled points, $\mathcal{S}$.
    \item The coordinator sets $w=\min\{d(u,v)\mid u,v\in \mathcal{S}\}$ and $\tau=\frac{2\alpha w}{5\epsilon}$.
    \item The coordinator runs Algorithm \ref{alg:iterative_greedy} on $\mathcal{S}$ and outputs the $k$ largest components $C_1',\dots,C_k'$ of $G_{\tau,b}$.
    \item If the total number of points in $C_1',\dots,C_k'$ is $\geq (1-b)n$ and for all $i$, $|C_i'|\geq 2bn$, then continue.
    Otherwise, increase $\tau$ to the smallest $\tau'>\tau$ such that $G_\tau\neq G_{\tau'}$, and go to the previous step.
    \item The coordinator creates a function $f:V\rightarrow [k]$ such that for all $v\in V$, $f(v)=\text{argmin}_{i\in [k]}d_{\text{med}}(v,C_i')$.
\end{enumerate}
\ttx{Output:}  Function $f:V\rightarrow [k]$ which defines a near-optimal clustering
\end{Frame}

\begin{proof}[Proof of Theorem \ref{thm:distr_sample}]
First we show that in step 3, the coordinator's set $\mathcal{S}$ of points is a uniformly random sample of the input of size $\frac{1}{\epsilon}\log 10k$.
Given $i$, given $v\in V_i$, the probability that $v\in\mathcal{S}$ is 
$\frac{1}{n_i}\cdot\frac{n_i}{n}\cdot\frac{1}{\epsilon}\log 10k=\frac{1}{\epsilon}\log 10k$.

Now we follow an analysis similar to \cite{balcan2009agnostic}.
Let $G_i$ denote the good points in $C_i\in\OPT$ and let $B$ denote the bad points in $\OPT$, as defined earlier.
Then since the clusters in $\OPT$ are large enough, we can use a similar reasoning as in Theorem \ref{thm:distr_as} to show that $|G_i|>5|B|$.
Furthermore, since our random sample is size $\Theta(\frac{1}{\epsilon}\ln \left(\frac{k}{\delta} \right))$, we can show that with probability at least
$1-\delta$, $|B\cap \mathcal{S}|<2(1+5/\alpha)\epsilon n$ and $|G_i\cap \mathcal{S}|\geq 4(1+5/\alpha)\epsilon n$, so $|G_i\cap \mathcal{S}|>2|B\cap \mathcal{S}|$ for all $i$.
Therefore, by running the first three steps of Algorithm \ref{alg:neighborhood}, we generate a clustering that is $O(\epsilon/\alpha)$-close to $\OPT$ on the sample.
So taking the largest connected components of this graph gives us a clustering that is $O(\epsilon/\alpha)$-close, restricted to $\mathcal{S}$.
If $w_{avg}$ is unknown, then we can apply a technique similar to Theorem \ref{thm:distr_as}.
Overall, we end up with a function $f$ defining a clustering with error $O(\epsilon)$ over all input points.

The communication complexity in the first two steps of the algorithm is $O(s\log n)$.
The third round communicates $\frac{1}{\epsilon}\log (10k)$ points, which uses $O(\frac{1}{\epsilon}\log k)$ bits of communication.
Therefore, the total communication is $O(s\log n+\frac{1}{\epsilon}\log k)$.

\end{proof}


\section{A Strong Notion of Stability}
\label{sec:separable}

Here we show that \emph{separation} is a strong and general notion of stability, that implies previously well-studied notions such as approximation stability and perturbation resilience. 

\vspace{0.1in}
\noindent \textbf{Lemma \ref{lem:reductions}.}\textit{(restated.)
Given $\alpha,\epsilon>0$, and a clustering objective (such as $k$-median), let $(V,d)$ denote 
a clustering instance which satisfies $c$-separation, for $c>(1+\alpha)n$ (where $n=|V|$).
Then the clustering instance also satisfies $(1+\alpha, \epsilon)$-approximation stability and $(1+\alpha)$-perturbation resilience.}
\begin{proof}
Given an instance $(V,d)$ that satisfies $c$-separation, first we prove this instance satisfies $(1+\alpha, \epsilon)$-approximation stability.
Consider a clustering $\mathcal{C}'$ of $(V,d)$ which is not equal to the optimal clustering. Then there must exist a point $p$ whose center
under $\mathcal{C}'$ is from a different optimal cluster.
Formally, there exist $p\in C_i^*$ and $q\in C_j^*$ such that $q$ is the
center for $p$ under $\mathcal{C}'$.
By definition of $c$-separation, we have $d(p,q)> (1+\alpha) n\cdot \max_i \max_{u,v\in C_i^*} d(u,v)$.
However, note that an upper bound on the optimal cluster cost is
$n \max_i \max_{u,v\in C_i^*} d(u,v)$.
Therefore, the cost of $\mathcal{C}'$ is at least a multiplicative $(1+\alpha)$ factor greater than the optimal clustering cost.
We have proven that any non-optimal clustering is not a $(1+\alpha)$
approximation, therefore, the instance satisfies $(1+\alpha,\epsilon)$-approximation stability.

Now we turn to perturbation resilience.
Assume we are given an arbitrary $(1+\alpha)$-perturbation of the metric $d$.
That is, we are given $d'$ such that for all $p,q\in V$, we have
$d(p,q)\leq d'(p,q)\leq (1+\alpha)\cdot d(p,q)$.
Then the optimal clustering is cost at most $(1+\alpha)\OPT$.
From the previous paragraph, any non-optimal clustering $\mathcal{C}'$ in $d$ must have cost
greater than $(1+\alpha)\OPT$, therefore, $\mathcal{C}'$ must have cost greater than $(1+\alpha)\OPT$ in $d'$.
It follows that the optimal clustering stays the same under $d'$,
and so the instance satisfies $(1+\alpha)$-perturbation resilience.
\end{proof}

\end{document}